\documentclass[sigconf]{acmart}
\usepackage{subcaption}
\usepackage{xcolor}
\usepackage{dblfloatfix}



\usepackage{amsmath}
\DeclareMathOperator*{\argmax}{arg\,max}
\newcommand{\abs}[1]{\left| #1 \right|}

\usepackage{algorithm}
\usepackage{algpseudocode}

\AtBeginDocument{%
  }

\setcopyright{acmlicensed}
\copyrightyear{2027}
\acmYear{2027}
\acmDOI{XXXXXXX.XXXXXXX}
\acmConference[ACM KDD '27]{Proceedings of the 33rd ACM SIGKDD Conference on Knowledge Discovery and Data Mining}{August 01--05,
  2027}{San Jose, CA}
\acmISBN{978-1-4503-XXXX-X/2018/06}

\begin{document}

\title{Adversarial Data Modeling in Epidemiology}

\author{Yiqi Su}
\email{yiqisu@vt.edu}
\affiliation{%
  \institution{Virginia Tech}
  \city{Alexandria}
  \state{VA}
  \country{USA}
}

\author{Christo K. Thomas}
\email{cthomas2@wpi.edu}
\affiliation{%
  \institution{Worcester Polytechnic Institute}
  \city{Worcester}
  \state{MA}
  \country{USA}
}

\author{Walid Saad}
\email{walids@vt.edu}
\affiliation{%
  \institution{Virginia Tech}
  \city{Alexandria}
  \state{VA}
  \country{USA}
}

\author{Sanmay Das}
\email{sanmay@vt.edu}
\affiliation{%
  \institution{Virginia Tech}
  \city{Alexandria}
  \state{VA}
  \country{USA}
}

\author{Bud Mishra}
\email{bud.mishra@gmail.com}
\affiliation{%
  \institution{New York University}
  \city{New York}
  \state{NY}
  \country{USA}
}

\author{Naren Ramakrishnan}
\email{naren@vt.edu}
\correspondingauthor
\affiliation{%
  \institution{Virginia Tech}
  \city{Alexandria}
  \state{VA}
  \country{USA}
}

\renewcommand{\shortauthors}{Su et al.}

\begin{abstract}
Epidemiological models increasingly rely on crowdsourced, self-reported behavioral data such as vaccination status, mask usage, and social distancing adherence. This data, however, is not passively sampled but instead strategically reported, making it a canonical case of adversarial input to a data mining pipeline. Individuals misreport for various reasons, e.g., to avoid penalties, to access benefits, or to express distrust in public health authorities. We introduce a data-modeling framework that casts the interaction between the population and a public health authority as a signaling game.
This approach provides both a generative model of strategically-corrupted behavioral data and a mechanism for the receiver to recover reliable signal from it. 
Individuals (senders) choose how to report their behaviors, while the public health authority (receiver) updates their epidemiological model(s) based on potentially distorted signals, and modifies its trust in incoming reports accordingly. Focusing on deception around masking and vaccination, we characterize analytically game equilibrium outcomes as distinct regimes of data corruption, and evaluate the degree to which deception can be tolerated while maintaining epidemic control through policy interventions.
In large scale simulations, our results show that even under pervasive dishonesty in pooling equilibria, well-designed sender and receiver strategies can still maintain effective epidemic control. Real-world validation further shows that behavioral distortions often exhibit structured patterns rather than arbitrary noise. This work advances the understanding of adversarial data in epidemiology and offers tools for designing more robust public health models in the presence of strategic user behavior.

\end{abstract}

\ccsdesc[500]{Computing methodologies~Probabilistic reasoning}
\ccsdesc[500]{Computing methodologies~Model development and analysis}
\ccsdesc[500]{Applied computing~Health informatics}


\maketitle

\section{Introduction}
Epidemiological modeling increasingly depends on large-scale, crowdsourced behavioral data to track and forecast disease transmission dynamics in real time \cite{banerjee2022,BergstromRS2002}. Beyond clinical surveillance and mobility records, self-reported behavioral information on individual vaccination compliance, mask usage, and adherence to social distancing has become central to capturing the effects of non-pharmaceutical interventions (NPIs) \cite{bff2025}. 

% \noindent 
However, empirical studies across diverse contexts demonstrate large and persistent gaps between self-reported and directly observed behaviors. For example, a national study in Kenya found that while only 12\% of individuals reported not wearing masks, direct observation revealed that nearly 90\% were non-compliant, a discrepancy exceeding 75 percentage points that was consistent across demographic groups \cite{kenya2021}. Similar survey-based studies in countries like the USA document substantial social-desirability and recall biases in self-reported masking and vaccination behavior, compounded by spatial heterogeneity, policy variation, and limited auditability of daily behaviors \cite{spatiotemporal}. While vaccination records are more structured and exhibit higher fidelity than other NPI measures, they nevertheless exhibit persistent misreporting rates of approximately 5–15\%, driven by reporting delays, incomplete registries, and misclassification errors~\cite{HURLEY2026128101,DALEY20242740,IRVING20096546}.

Moreover, a persistent and underexamined challenge lies in the strategic nature of self-reported data. Individuals may misreport or withhold behavioral information to avoid penalties, retain access to workplaces or social benefits, or due to distrust in public health authorities (PHA) ~\cite{levy2022,shiman2023}.  
In such settings, reported data 
can no longer be treated as passively sampled observations, but instead constitute strategic signals.
This misalignment between reported and actual behavior can introduce systematic underestimation of key parameters and impair NPI effectiveness, as demonstrated by recent studies~\cite{WuNature2020,PullanoNature2021,MilanesiArxiv2023}. Unfortunately, state-of-the-art models largely fail to account for strategic misreporting by the population~\cite{PatelJTB2005} %, TonkensECC2021} 
and the few studies that do acknowledge parameter uncertainty
(e.g.,~\cite{TannerMB2018}) 
view it solely as an estimation problem rather than as a dynamic behavioral process that can be explicitly modeled.

% In such regimes, studying the impact of strategical misreporting is critical for understanding how much distorted behavioral data can be tolerated before epidemic control degrades, and for designing robust public health decision frameworks that explicitly account for strategic or socially driven misreporting rather than assuming truthful compliance.

To address these issues, we propose a game-theoretic framework that explicitly models the interaction between individuals and a cognizant public health authority (hereinafter, referred to as PHA) as a \emph{signaling game}~\cite{SobelGT2020}. In our formulation, each individual (sender) chooses whether and how to report their health behavior (e.g., vaccination, masking), while the PHA (receiver) interprets these signals to estimate key epidemiological parameters and adjusts its recommendations accordingly. Crucially, both players act under bounded rationality and operate with partial information, yielding a dynamic interplay between strategic communication and epidemic control. This approach integrates rational speech act (RSA) models~\cite{DegenARL2023}, which formalize communication as recursive reasoning between senders and receivers with epidemiological compartmental dynamics, offering a principled way to integrate behavioral information into epidemic forecasting.  Our key contributions are:
\begin{itemize}
    \item \textbf{Game-theoretic epidemic model.} We couple a stochastic compartmental model with a two-player signaling game, allowing the PHA to infer credibility and account for strategic misreporting rather than assuming truthful reporting, and to design policies that remain effective even when behavioral observations are adversarial.
    
    \item \textbf{Equilibria characterization and policy design.} 
    By characterizing the signaling-game equilibria, we identify when truthful reporting can arise and how policy levers can be tuned to deter strategic misreporting while preserving predictive fidelity. Building on these insights, we design an adaptive feedback policy that weights reported data by inferred credibility to maintain the controlled reproduction number $R_c \leq 1$.
    
    % \item \textbf{Quantitative impact of misreporting.} We characterize through Monte Carlo experiments how much deception can be tolerated and demonstrate that there exists a tolerance frontier specifying how much misreporting can be absorbed before control is lost.

    \item \textbf{Generalizable simulation pipeline.} Rather than focusing on epidemic forecasting, we develop a generalizable modeling and simulation framework that integrates signaling-game equilibria with epidemiological dynamics to study how behavioral deception interacts with policy adaptation and epidemic control. The framework enables systematic evaluation of intervention strategies under varying behavioral and epidemiological assumptions, advancing the study of adversarial data in epidemiology and supporting robust public-health decision-making in the presence of strategic user behavior.

\end{itemize}

\section{Related Work}
\paragraph{Compartmental epidemiological models.} Classical models such as SEIR have been extended with compartments for vaccination ($V$) and asymptomatic transmission ($A$), yielding SVEAIR models effective at reproducing COVID-19 dynamics~\cite{ChoiJTB2020}.
Recent advances embed epidemiological models in neural ODE/SDE frameworks with transmission graphs and uncertainty-aware control~\cite{thiagarajan2022machine,EARTH2025,nSDE2025}.
Our approach builds on this line of research, adding process noise, weekly time–steps, and explicit feedback from reported behavior to policy, thereby closing the loop between epidemiology and decision making.

\paragraph{Measurement error and strategic misreporting.}
Self–reported health data suffer from both random noise and strategic bias.  
Empirical research demonstrates that self-reported health behaviors such as vaccination status, mask-wearing compliance, and infection history are frequently misreported, especially when disclosure could affect employment, travel, or social standing~\cite{Battaglia1997,Merrell2024}.
Statistical corrections (e.g.\ weighting, imputation) address random error but struggle with incentive-driven deception. Our work complements these efforts by modeling misreporting endogenously and quantifying how it degrades policy performance across incentive regimes.

\paragraph{Signaling games and equilibrium selection.}
Information asymmetries in public health have often been modeled as signaling or screening games where individuals send costly messages and the PHA responds with incentives or restrictions.
Equilibrium refinements (separating, pooling, mixed) determine whether truth telling, deception, or mixed strategies persist.
Beyond traditional one–shot or static analyses, recent work has examined, in other domains, strategic disclosure under multiple alternatives~\cite{AAAI2024Disclosure} and learning-induced dynamics that deviate from Nash equilibria~\cite{IJCAI2023MultiMemory}.
In our work,
we study the joint evolution of equilibrium strategies, epidemic dynamics, and responsive policy measures.

\begin{figure}%[htbp!] 
    \centering
    \includegraphics[width=.9\columnwidth]{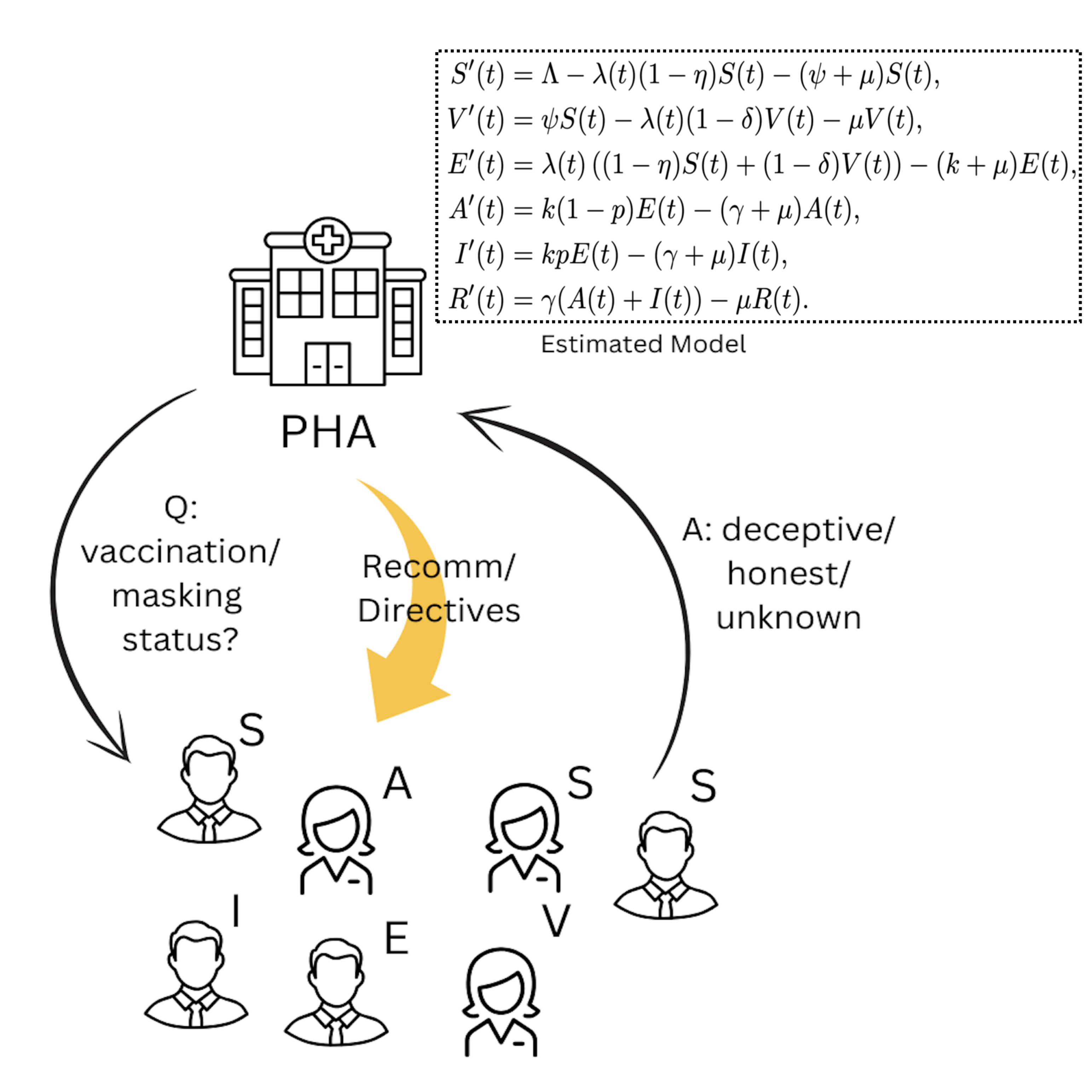}
    % \vspace{-4mm}
    \caption{Illustration of the signaling game framework.}
    \label{fig:System_Model} % \vspace{-4mm}
\end{figure}

\section{Framework}
\paragraph{Basics.} We model interactions between the PHA and individuals as a two-way conversation, as illustrated in Figure~\ref{fig:System_Model}.
In this model, the PHA queries individuals about their vaccination and masking status. A subset of people respond, with their answers reflecting individual behaviors (honesty or deception). Besides this information (which is not necessarily accurate), the
only ground truth information available to the PHA is the number of hospitalizations (which is typically obtained by aggregating across all clinics and hospitals in its jurisdiction). It is reasonable to assume that the number of people hospitalized is equal to a fraction of the number of people infected by the virus~\cite{FoxPNAS2022}. Details of such estimation and modeling of other compartments in the SVEAIR model are covered later.
Based on the response from population about their masking and vaccination, the PHA updates its epidemiological model and provides recommendations to the public on protective strategies to adopt.

\paragraph{Epidemiological model.}
We cast the interaction between the population and the PHA as a two-player signaling game $\mathcal{G}$, in which individual behaviors (types), reported responses (messages), and PHA strategies (actions) are drawn from discrete sets. Let $t$ denote time. 
In $\mathcal{G}$, the sender $T$ is a group of individuals $\mathcal{K}$, comprising $K(t)$ members, each in one of several epidemiological states: susceptible $(S)$, vaccinated $(V)$, exposed $(E)$, symptomatic $(I)$, asymptomatic $(A)$, or recovered $(R)$. Individuals send messages $m \in \mathcal{M}$ to the receiver $R$, with $\mathcal{M}$ denoting the feasible set of messages. The receiver $R$ represents the PHA’s epidemiological model, which tracks the evolution of these states through a system of ODEs~\cite{ChoiJTB2020}. 
The epidemiological model, essentially a variant of the SVEAIR approach, is
described in Appendix~\ref{supp:epi_model} and
illustrated in 
Figure~\ref{fig:sveair}. 
It is important to note the vaccination rate $\psi$ and the masking rate $\eta$ in this model, both of which reduce the contact rate.

%Next, we discuss in detail the specifics of signaling game. For notational convenience, we may omit the time index $t$; however, time dependency is understood to be implicit in the formulation.

%structure described below thus forms the foundation of the signaling game, whose strategic implications will be elaborated in the next section.

\begin{figure}[t]
\centering
\includegraphics[width=.95\columnwidth]{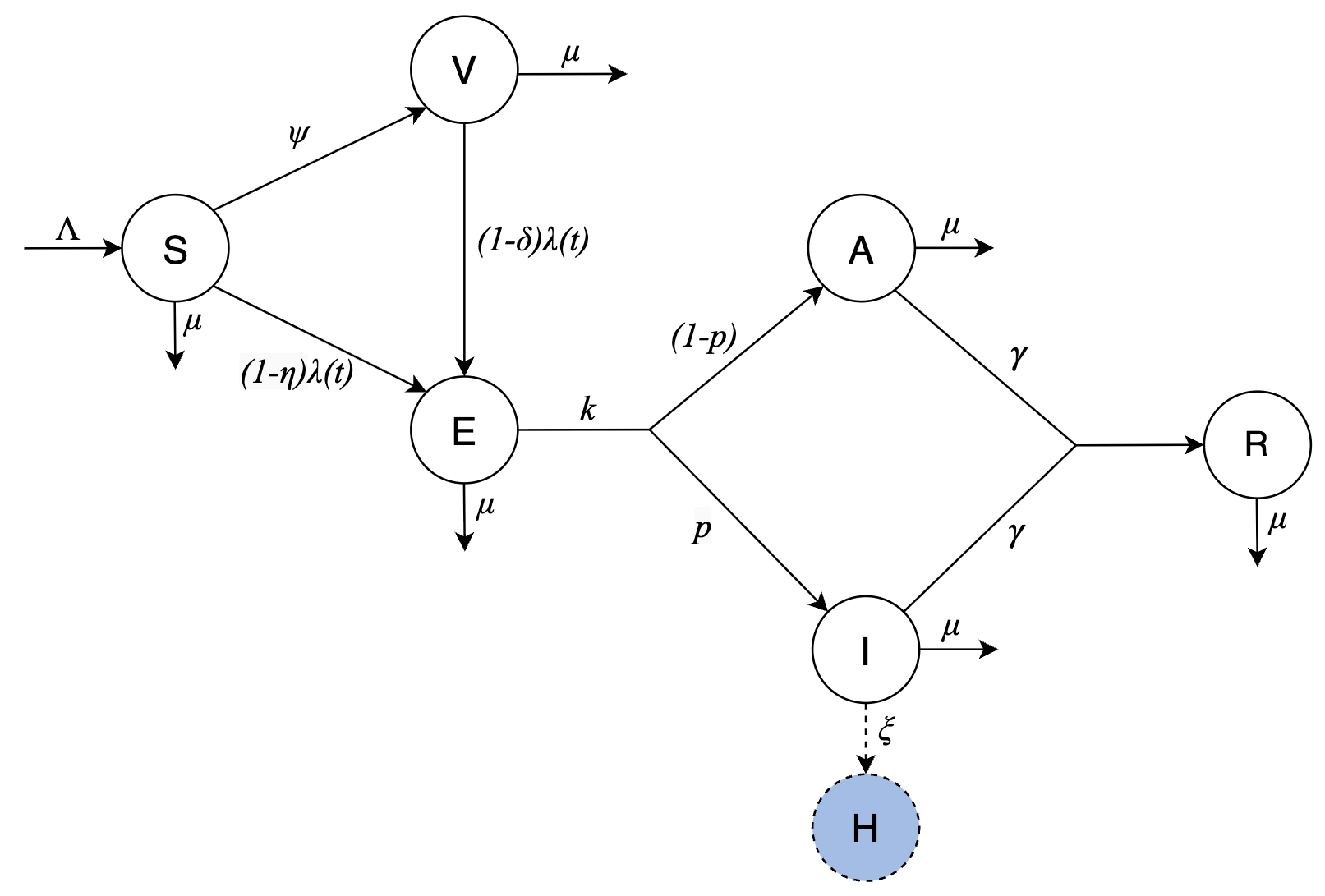} % Slightly narrower than full column
% \vspace{-2mm}
\caption{Compartmental SVEAIR epidemiological model with vaccination and masking. The `H' is not a separate compartment but denotes (observed) hospitalization numbers, typically a fraction of the true numbers from the infected compartment.}
\label{fig:sveair}
% \vspace{-4mm}
\end{figure}

\paragraph{Realistic assumptions for vaccination and masking.}
In the weekly interactive loop between individuals and the PHA (see details in Appendix~\ref{supp:weekly-loop}), masking compliance is defined based on self-reported frequency over the prior 5–7 days, classifying “all of the time” or “most of the time” as masking and other responses as non-masking, consistent with large-scale survey practice~\cite{spatiotemporal}. To account for the distinctions of vaccination and masking behaviours, the deception rates are pre-specified as 5-15\% and 40-80\%, respectively~\cite{kenya2021,Archambault2023}. 

\paragraph{Signaling game.}
%For notational convenience, we may omit the time index $t$; however, time dependency is understood to be implicit in the formulation.

For the signaling game $\mathcal{G}$, individuals communicate messages from a set $\mathcal{M}$, representing vaccination and masking status.
We encode messages as two bits $b_1b_2$, where $b_1=0$ ($1$) indicates ``Vaccinated'' (``Unvaccinated''), and $b_2=0$ ($1$) indicates ``Masking'' (``Not Masking'').
Thus, the message space is $\mathcal{M}=\{00,01,10,11\}$.
Note that $0$ denotes desirable behavior (i.e., desired by the PHA) on the part of individuals.
We define an individual's behavioral type as $\theta_i=(\theta_i^v,\theta_i^m)\in\Theta$, where $\Theta=\{00,01,10,11\}$, $\theta_i^v$ denotes the true vaccination status, and $\theta_i^m$ denotes the true masking status. For example, $\theta_i=01$ indicates that an individual is vaccinated but does not mask.

Reported messages may, of course, not match actual behavior, e.g., an unvaccinated or non-masking individual may claim otherwise, allowing the
capture of deception. 
The sender strategy is therefore represented by $g(m_i\mid \theta_i)$, which denotes the probability that an individual with true behavioral type $\theta_i$ communicates message $m_i$.
For truthful individuals, $g(m_i=\theta_i\mid\theta_i)=1$, while deceptive individuals may choose messages that provide higher individual utility.
%While we restrict attention to vaccination and masking, the framework generalizes to richer behavioral states.

% Futher, we assume that only $M \leq K$ individuals in group $\mathcal{K}$ report their status. Individual types are represented by $\boldsymbol{c} = [c_1,\cdots,c_K]^T$, where $c_i=1$ denotes honesty and $c_i=0$ indicates deception. Let $g(m\mid c)$ denote the probability of sending message $m$ given type $c$. Since persistent dishonesty reduces payoffs (via penalties and social costs), deceptive individuals adopt probabilistic strategies balancing truth and misreporting. The receiver, in turn, forms beliefs about $\boldsymbol{c}$ to update the epidemiological model (Appendix~\ref{supp:epi_model}), where $\widehat{\psi}$ and $\widehat{\eta}$ denote estimated vaccination and contact rates, respectively. We now define sender and receiver payoffs.
We assume that only a subset of individuals $\mathcal{K}_c\subseteq\mathcal{K}$ responds to PHA survey, while the remaining individuals $\mathcal{K}_s=\mathcal{K}\setminus\mathcal{K}_c$ do not provide behavioral information. Then PHA forms posterior beliefs over behavioral types $q(\theta_i\mid m_i)$, which are used to estimate population-level behavioral parameters, including the vaccination rate $\widehat{\psi}$ and masking rate $\widehat{\eta}$.

\paragraph{Sender utility function.} Falsely claiming to be vaccinated may offer incentives, such as earning more in the workplace. Similarly, falsely stating compliance with mask-wearing may help avoid penalties imposed by the PHA. In short, while honesty about individual actions may support the societal goal of achieving zero infection rates, it could negatively impact individual payoffs. Let $I_v$ and $I_m$ denote the incentives associated with falsely claiming vaccination and masking compliance, respectively.
The sender utility is defined according to the reported message:
\begin{equation}
U_s = \sum_{k\in\mathcal K} U_s^{(k)}(\theta_k,m_k),
\label{eq:utility}
\end{equation}
where the individual-level utility is given by
% % \begin{equation}
% % \begin{aligned}
% % U_s^{(k)}(c_k, m_k) &= g\big(m_k = \text{``00''} \mid c_k = 0\big)(I_v + I_m) \\
% % & +\; g\big(m_k = \text{``01''} \mid c_k = 0\big) I_v \\
% % & +\; g\big(m_k = \text{``10''} \mid c_k = 0\big) I_m.
% % \end{aligned}
% % \label{eq_sender_payoff}
% % \end{equation}

% \begin{equation}
% \begin{aligned}
% U_s^{(k)}(\theta_k,m_k)
% =& 
% g(m_k=00\mid\theta_k)
% \left(
% \mathbf{1}(\theta_k^v=1)I_v + \mathbf{1}(\theta_k^m=1)I_m
% \right)
% \\
% &+
% g(m_k=01\mid\theta_k)
% \mathbf{1}(\theta_k^v=1)I_v
% \\
% &+
% g(m_k=10\mid\theta_k)
% \mathbf{1}(\theta_k^m=1)I_m .
% \end{aligned}
% \label{eq_sender_payoff}
% \end{equation}

\begin{equation}
\begin{aligned}
U_s^{(k)}(\theta_k) = 
\sum_{m_k\in\mathcal M}
g(m_k\mid\theta_k)
\Delta(\theta_k,m_k),
\end{aligned}
\label{eq_sender_payoff_modified}
\end{equation}

and
\begin{equation}
\begin{aligned}
\Delta(\theta_k,m_k)
=&
I_v \mathbf{1} (\theta_k^v=1, m_k^v=0) \\
& + I_m \mathbf{1}(\theta_k^m=1, m_k^m=0)
\end{aligned}
\label{eq:lie_gain}
\end{equation}
denotes the incentive obtained from falsely claiming compliance.

% \begin{equation}
% \begin{aligned}
% U_s^{(k)}(\theta_k)
% =
% \sum_{m_k\in\mathcal M}
% g(m_k\mid\theta_k)
% \Big[
% &I_v\mathbf{1}(\theta_k^v=1,m_k^v=0)
% \\
% &+
% I_m\mathbf{1}(\theta_k^m=1,m_k^m=0)
% \Big].
% \end{aligned}
% \label{eq_sender_payoff}
% \end{equation}

% \noindent
Here, to compute the utility, which is the incentive upon deceiving, we consider only individuals who are either not vaccinated or have not followed masking guidelines. For instance, \( g(m_k = ``00" \mid \theta_k=``11") \) denotes a scenario where the individual has neither adopted vaccination nor masking, yet falsely communicates that they have. 
%Similarly, \( g(m_k = ``01" \mid c_k=0) \) and \( g(m_k = ``10" \mid c_k=0) \) denotes the scenarios where the individual has not adopted vaccination and masking, respectively, yet falsely communicates that they have.

Maximizing only $U_s^{(k)}(\theta_k)$ leads to a trivial solution where the individuals always choose to be dishonest about their strategies, since that will bring them maximum payoffs. %However, such a strategy 
This leads to increased infection rates, closing down of businesses, and potentially punitive measures by the PHA.
Therefore, the sender in our model is interpreted as \emph{a representative population agent} rather than single strategic individuals. The population consists of a continuum of individuals facing identical incentives but potentially making heterogeneous reporting choices. Accordingly, the sender’s strategy $g(m \mid \theta)$ represents a population-level mixed strategy, i.e., the empirical distribution of reports generated by individuals under the same incentive structure. In this way, the sender communication strategies must represent a balance between maximizing $U_s^{(k)}(\theta_k)$ and minimizing these losses. 

We represent the losses as proportional to the negative of semantic accuracy, which is defined as the accuracy with which the receiver interprets the individual types from the messages communicated and written as:
\begin{equation}
U_L(p, q) = -\frac{1}{K} \sum\limits_{i \in \mathcal{K}} \sum\limits_{j \in \mathcal{M}} p_{ij} q_{ji},
\end{equation}
where $p_{ij} = g(m_j \mid \theta_i)$, representing the sender communication strategy, and $q_{ji} = q(\theta_i \mid m_j)$, representing the receiver posterior belief. For non-responsive individuals, we set $p_{ij} = q_{ji} = 0$.
The sender's objective is to find a communication policy $g(m_k \mid \theta_k)$ by maximizing the following objective:
% \vspace{-2mm}

\begin{equation}
\mathcal{P}_1:\left[g(m_k \mid \theta_k)^{\ast}\right]_{\forall k} = \argmax\limits_{g(m_k \mid \theta_k),\, \forall k} \left( U_s - \lambda_1 U_L \right),
\end{equation}
where $\lambda_1$ is a weighting factor that represents an optimal trade-off between maximizing $U_s$ and minimizing $U_L$.

Next, to  understand  the   implications of user    misreporting, we define the control reproduction number $R_c$ as the rate at which a person can pass on the infection. The expression for $R_c$ follows  directly from \cite{ChoiJTB2020}, as
\begin{equation}
    R_c(\psi,\eta) = \underbrace{\frac{\beta k}{(k+\mu)(\gamma + \mu)}}_{R_0}\frac{(1-\delta)\psi + (1-\eta)\mu}{\psi + \mu}.
\label{eq_Rc}
\end{equation} 
A value of $R_c(\psi,\eta)>1$ represents endemic equilibrium, where the disease remains endemic in the population. Here $\beta = \beta_0(p+b(1-p))$  indicates the weighted transmission rate considering both symptomatic and asymptomatic
cases. Eq.~\eqref{eq_Rc} means that as  $R_c(\psi,\eta)$ increases, it may have severe consequences of the functioning of business, schools and other institutions. This will adversely impact the payoffs of each individual, in terms of loss of job or restrictions to the number of employees that can work, thereby reducing their salary or benefits. In practice, the reproduction number can be estimated using the available observations  \cite{ShiArxiv2025} by the PHA, such as hospitalization. 
To capture this effect, we modify the sender utility using a discount factor
\begin{equation}
\begin{aligned}
\widetilde{U}_s^{(k)}(\theta_k)
=
e^{-aR_c(\psi,\eta)}
U_s^{(k)}(\theta_k).
\end{aligned}
\label{eq_sender_payoff_modified}
\end{equation}
where $a$ is a constant economic coefficient that controls the slope of  the exponential  curve.
% \begin{equation}
%     \begin{aligned}
%     U_s^{(k)}(\theta_k,m_k) &= g(m_k = ``11"\mid \theta_k) e^{-aR_c(\psi,\eta)}(I_v+I_m) \\ & + g(m_k = ``10"\mid \theta_k)e^{-aR_c(\psi,\eta)}I_{v} \\ &+ g(m_k = ``01"\mid \theta_k) e^{-aR_c(\psi,\eta)}I_m,
%     \label{eq_sender_payoff_modified}
%     \end{aligned}
% \end{equation}
% where $a$ is the economic factor, which is a constant coefficient that controls the slope of  the exponential  curve.
As per Eq.~\eqref{eq_sender_payoff_modified}, higher $R_c(\psi,\eta)$ lowers individual payoffs. Stricter PHA measures thus reflect both rising infection rates and their economic impact on daily life. The modified sender payoffs ensure that sender utility depends not only on individual strategies but also on the PHA’s ability to predict and control the epidemic without disrupting normal activities.

%As per \eqref{eq_sender_payoff_modified}, as $R_c(\psi,\eta)$ increases, the individual payoffs reduce. Hence, when the PHA introduces stricter measures, that is indicative of the economic impact due to an exploding infection rate and hence on the day-to-day normal activities of the population. The modified sender payoffs ensure that $U_s^{(k)}$ depends not just on individual strategies but also on whether the PHA is able to accurately predict the evolution of epidemiology and hence bring it under control, without impacting the day to day activities of the population.

% \begin{figure}[t]
% \centering
% \includegraphics[width=.8\columnwidth]{AnonymousSubmission/LaTeX/sg.png}
% \caption{Illustration of signaling games.}
% \label{fig:sg}\vspace{-3mm}
% \end{figure}

\paragraph{Receiver utility function.}
The receiver's objective here is to model accurately the epidemiological ODEs in \eqref{eq_ODE} using the estimated values of $\boldsymbol{\phi}(t) = \{\widehat{\psi}(t), \widehat{\eta}(t), \widehat{\lambda}(t), \widehat{\gamma}(t)$. The estimated ODE by the PHA has the same structure as the ground-truth epidemiological model given in Appendix~\ref{supp:epi_model}. The estimation error will typically be given by $\abs{X(t) - \widehat{X}(t)}^2$ for $X \in {S,V,E,I,A,R}$. However, the PHA cannot compute this error directly due to lack of access to ground-truth ODEs (as noted earlier). Fortunately, it can estimate hospitalizations, which are assumed to be proportional to infections~\cite{FoxPNAS2022}. We therefore model the hospitalization count as:
$H(t) = \xi I(t)$,
where $\xi$ is the infection–hospitalization ratio (IHR).
% The receiver's objective is to minimize this error and thereby ensure that the protection strategies followed by it are enough to control the disease. 
For a given sender type and belief,  the hospitalization prediction error (also called distortion) can be written as:
% \begin{equation}
%     D(t) = \sum\limits_{\boldsymbol{c}} p(\boldsymbol{c} \mid \boldsymbol{m}(t)) 
%     \sum\limits_{\boldsymbol{\theta}(t)} p(\boldsymbol{\theta}(t) \mid \boldsymbol{c}) 
%     \left| H(t) - \widehat{H}(t) \right|^2.\label{eq_error}
% \end{equation}

\begin{equation}
D(t)
=
\sum_{\theta}
q(\theta\mid m(t))
\sum_{\boldsymbol{\phi}(t)}
q(\boldsymbol{\phi}(t)\mid\theta)
\left|
H(t)-\widehat H(t)
\right|^2 .
\label{eq_error}
\end{equation}

This distortion measures the inconsistency between observed hospitalization dynamics and the epidemiological trajectory inferred
from behavioral reports, with
$q(\boldsymbol{\phi}(t)\mid\theta)$ capturing the uncertainty in
epidemiological parameters associated with behavioral type $\theta$.
The receiver seeks beliefs that reduce uncertainty about population
behavior while maintaining bounded epidemiological distortion. Therefore,
the receiver optimization problem is formulated as

\begin{equation}
\begin{aligned}
\mathcal{P}_2:
\quad
\min_{q(\theta\mid m)}
&
-
\sum_{\theta,m}
P(m)
q(\theta\mid m)
\log q(\theta\mid m)
\\
\text{s.t.}
&
\sum_{\theta,m}
P(m)
q(\theta\mid m)
D(\theta)
\leq D^*,
\\
&
\sum_{\theta}
q(\theta\mid m)=1,
\quad
\forall m\in\mathcal M .
\end{aligned}
\label{eq:P2}
\end{equation}

where $P(m)$ denotes the empirical distribution of observed messages received by the PHA, and $D^*$ represents the maximum tolerable
hospitalization prediction distortion.
Based on the inferred behavioral parameters and the estimated epidemiological dynamics, the PHA updates its intervention recommendations through adaptive correction terms:

\begin{align}
\psi_r(t)
&=
\widehat{\psi}(t)
+
\widetilde{\psi}(t),
\\
\eta_r(t)
&=
\widehat{\eta}(t)
+
\widetilde{\eta}(t).
\label{eq:combined_params}
\end{align}

Here, $\widehat{\psi}(t)$ and $\widehat{\eta}(t)$ represent the vaccination and masking levels inferred from behavioral reports,
respectively. The correction terms
$\widetilde{\psi}(t)$ and $\widetilde{\eta}(t)$ compensate for inference errors by reducing the hospitalization prediction distortion. The
resulting quantities $\psi_r(t)$ and $\eta_r(t)$ are the intervention recommendations issued by the PHA and are used to drive the subsequent epidemic evolution.

To characterize the intervention level required for epidemic control, we derive the herd-immunity vaccination threshold from the controlled reproduction number. Specifically, epidemic control requires $R_c(\psi,\eta)<1$.
Solving Eq.~\eqref{eq_Rc} for the vaccination level yields

\begin{equation}
\psi_{HI}
=
\frac{
\mu
\left(
R_0(1-\eta)-1
\right)
}
{
1-R_0(1-\delta)
},
\label{eq:psi_HI}
\end{equation}

where $1<R_0<\frac{1}{1-\delta}$.
The threshold $\psi_{HI}$ represents the minimum vaccination level
required to suppress transmission under a given masking level $\eta$.

\paragraph{Game solution.}

%Shannon's rate-distortion theory \cite{CoverIT1991} aims to compute the solution to the following problem, which is the transmitter problem here:
%\begin{equation}
%    \mathcal{P}_1: R^{\ast}(D) = \min\limits_{-U\leq -U^{\ast}} \mathbb{I}(\boldsymbol{c};\boldsymbol{m})
%\end{equation}

Given the sender communication strategy, the receiver updates its belief about behavioral types using Bayesian inference. Following the rational speech act (RSA) framework, the receiver posterior is given by
\begin{equation} P_L(\theta\mid m) \propto P_S(m\mid\theta)q(\theta), \label{eq:rsa_listener} 
\end{equation} 
where $q(\theta)$ denotes the prior distribution over behavioral types and $P_S(m\mid\theta)$ represents the sender communication strategy. 
Conversely, given the receiver belief, a pragmatic sender chooses messages according to

\begin{equation} 
P_S(m\mid\theta) \propto e^ {\left( \tau \left(u_s(\theta,m) - \lambda_1 u_L(\theta,m) \right) \right)}, 
\label{eq:rsa_speaker} 
\end{equation} 
where $\tau$ controls the rationality of the sender, $u_s(\theta,m)$ denotes the message-level sender utility, and $u_L(\theta,m)$ denotes the semantic loss associated with inferring behavioral type $\theta$ from message $m$.

Given the sender strategy, the receiver posterior is obtained by solving
the constrained inference problem in Eq.~\eqref{eq:P2}. The receiver balances two objectives, i.e., reducing uncertainty in inferred behavioral states and maintaining consistency with observed epidemic dynamics.
The corresponding Lagrangian formulation is

\begin{equation}
\begin{aligned}
\mathcal{L}
=&
-\sum_{\theta,m}
P(m)P_L(\theta\mid m)
\log P_L(\theta\mid m)
\\
&
+\lambda_2
\left(
D(P_L)-D^*
\right),
\end{aligned}
\label{eq:Lagrangian_receiver}
\end{equation}

where $\lambda_2$ is the Lagrange multiplier controlling the trade-off between inference confidence and epidemiological distortion.
Solving this constrained optimization yields the receiver posterior

\begin{equation}
P_L(\theta\mid m)
=
\frac{
P_S(m\mid\theta)P(\theta)
\exp(-\lambda_2D(\theta))
}
{
\sum_{\theta'}
P_S(m\mid\theta')
P(\theta')
\exp(-\lambda_2D(\theta'))
},
\label{eq:posterior_solution}
\end{equation}

where $D(\theta)$ denotes the expected hospitalization prediction distortion associated with behavioral type $\theta$. This posterior preserves information contained in the received message
while down-weighting behavioral types whose implied epidemic dynamics
are inconsistent with observed hospitalization trajectories.

\section{Equilibrium Analysis}

%\subsubsection*{Scenario 1 - %$R_c(\psi,\eta) < 1$}

%\subsubsection*{Scenario %2-$R_c(\psi,\eta) > 1$} 
%\subsection*{Perfect Bayesian Nash Equilibrium}

%We define the signaling game as the tuple $\left(\mathcal{T},\mathcal{M},\mathcal{Q}\right)$, where $\mathcal{T}$ is the set of user types, $\mathcal{M}$ is the sender strategy space where $g(m\mid t)$ lies and $\mathcal{Q}$ is the strategy space of receiver where the learned distributions $\mathcal{Q}$ lies. A perfect Bayesian Nash equilibrium (NE) of the signaling game $\mathcal{G}$ is a strategy profile $()$ and posterior beliefs of the receiver about individual types $q(\boldsymbol{t})$ such that:
%\begin{equation}
%\begin{align}
%g^{\ast}(m\mid \boldsymbol{c}) & \in E_s \\
%\{\tilde{\psi},\tilde{\eta},p(\boldsymbol{\theta}(t)\mid \boldsymbol{m}(t))\} &\in \max\limits_{q}\,\,  %E_v(q_v) + E_{nv}(q_{nv})+ E_m(q_m) \\ &  E_{nm}(q_{nm}) +\alpha \mathbb{I}(\boldsymbol{\theta}(t);\boldsymbol{m}(t)) \\
%\mbox{s.t.}\,\, & \bar{e}(t)\leq %\epsilon
%\end{align}
%\end{equation}

\paragraph{Pooling equilibrium.}

% A pooling equilibrium~\cite{HutteggerPNAS2014} is achieved when all the non-compliant are dishonest. In this case, the messages communicated do not convey any information about the strategies of the population. The PHA will have  to purely rely  on its  prior  information,  which  is  the hospitalization   number  $H(t)$.

A pooling equilibrium~\cite{HutteggerPNAS2014} is achieved when all the non-compliant individuals are dishonest. In this case, the messages communicated do not convey any information about the true behavioral strategies of the population, and the PHA must rely solely on prior information and aggregate observations, such as the hospitalization signal $H(t)$.

\begin{theorem}
\label{thm:pooling}
Suppose all non-compliant individuals strictly prefer misreporting over truthful reporting for both vaccination and masking, i.e.,
\[
U_s^{\text{lie}} > U_s^{\text{truth}}
\]
for all non-compliant types. Then there exists a pooling equilibrium in which all individuals report full compliance, i.e., $m_k = 00$ for all $k\in\mathcal{K}_c$, and the receiver’s posterior beliefs coincide with the prior distribution over types.
\end{theorem}

% In this equilibrium, reported behavioral data are uninformative, and epidemic control relies entirely on the PHA’s inference from hospitalization dynamics rather than from self-reported signals.

\paragraph{Separating equilibrium.}
A separating equilibrium \cite{BergstromRS2002} is achieved when all non-compliant individuals are honest. In this case, the reported messages perfectly reveal individuals’ true behaviors, enabling the PHA to accurately infer risk and select interventions such that $R_c$ falls below one, with the induced policy intensity $\widehat{\psi}$ exceeding the herd-immunity threshold $\psi_{HI}$.

\begin{theorem}
\label{thm:separating}
If all responsive individuals report truthfully, i.e., $m_k = \theta_k$ for all $k\in\mathcal{K}_c$, then the signaling game admits a separating equilibrium. In this equilibrium, the PHA recovers unbiased estimates of behavioral parameters $(\psi,\eta)$ and can select policies such that
\[
R_c(\psi,\eta) < 1 \quad \text{and} \quad \widehat{\psi} > \psi_{HI},
\]
ensuring convergence to a disease-free equilibrium.
\end{theorem}

% Separating equilibria represent an idealized information regime in which signaling inefficiency does not distort epidemic control.

\paragraph{Partial pooling equilibrium.}
Departing from the extreme pooling equilibrium and the idealized separating equilibrium, a partial pooling equilibrium arises when some individuals report truthfully while others engage in deception. This is captured by the mixing factor $\alpha$. In this case, although self-reported signals are imperfect, the PHA may still extract partial information and attempt to enforce policies satisfying $R_c(\psi,\eta) < 1$ and $\widehat{\psi} > \psi_{HI}$.

\begin{theorem}
\label{thm:partial-pooling}
Consider the signaling game $\mathcal{G}$ with sender utility weight $\lambda_1$ on semantic accuracy, masking incentive $I_m$, economic factor $a$, and population type distribution $(\pi_{00}, \pi_{01}, \pi_{10}, \pi_{11})$. Define the controlled reproduction numbers $R_c(0)$ under full separation and $R_c(1)$ under full pooling, with sensitivity $\beta = \frac{dR_c}{d\alpha}\big|_{\alpha=0} > 0$.

\begin{enumerate}
    \item \textbf{Existence.} A partial pooling equilibrium with mixing probability $\alpha^* \in (0,1)$ exists if and only if the semantic weight satisfies:
    \begin{equation}
        e^{-aR_c(1)} I_m < \lambda_1 < e^{-aR_c(0)} I_m + a\beta e^{-aR_c(0)} I_m.
        \label{eq:existence-condition}
    \end{equation}

    \item \textbf{Characterization.} When the existence condition holds, the equilibrium mixing probability $\alpha^*$ satisfies the fixed-point equation:
    \begin{equation}
        \alpha^* = \frac{e^{-aR_c(\alpha^*)} I_m}{\lambda_1 \left(1 - q(\theta_{01} \mid m = 00)\right)},
        \label{eq:fixed-point}
    \end{equation}
    where $q(\theta_{01} \mid m = 00) = \frac{\alpha^* \pi_{01}}{\pi_{00} + \alpha^*(\pi_{01} + \pi_{10}) + \pi_{11}}$ is the posterior belief that a pooling message originated from type $\theta_{01}$.
    
    \item \textbf{Approximation.} For small deviations from separation (small $\alpha$), the equilibrium mixing probability admits the closed-form approximation:
    \begin{equation}
        \alpha^* \approx \frac{e^{-aR_c(0)} I_m}{\lambda_1 + a\beta e^{-aR_c(0)} I_m}.
        \label{eq:alpha-approx}
    \end{equation}
\end{enumerate}
\end{theorem}

\begin{proof}
See Appendix~\ref{supp:equilibrium-pooling}.
\end{proof}
The bound in  \eqref{eq:existence-condition} reveals that partial pooling emerges only when the semantic accuracy weight $\lambda_1$ lies in an intermediate range: for small $\lambda_1$ incentives dominate and all non-compliant individuals pool ($\alpha=1$); for large $\lambda_1$ accuracy dominates and all individuals separate ($\alpha=0$). %If $\lambda_1$ falls below the lower bound, incentives dominate and all non-compliant individuals pool ($\alpha = 1$). If $\lambda_1$ exceeds the upper bound, accuracy concerns dominate and all individuals separate ($\alpha = 0$). 

We decompose the population $\mathcal{K}$ into $\mathcal{K}_c$, the set of individuals responding to PHA queries, and $\mathcal{K}_s$, the set of silent individuals. The aggregate deception level $\tilde m\in[0,1]$ is defined as:
\begin{equation}
\begin{aligned}
    \tilde{m} &= \frac{1}{|\mathcal K|}
\left[\frac{1}{2}\sum\limits_{k\in \mathcal{K}_c}\sum\limits_{m_k} g(m_k\mid \theta_k) \lVert m_k-m_k^{(\ast)} \rVert_H + \frac{1}{2}\sum\limits_{k\in \mathcal{K}_s} 2\right] \\ 
              &= \frac{1}{|\mathcal K|}
\left[\frac{1}{2}\sum\limits_{k\in \mathcal{K}_c}\sum\limits_{m_k} g(m_k\mid \theta_k) \lVert m_k-m_k^{(\ast)} \rVert_H + \abs{\mathcal{K}_s}\right],
\end{aligned}
\end{equation}

where $\lVert m_k - m_k^{(\ast)} \rVert_H$ denotes the Hamming distance between the reported message and the ground-truth behavior.
Using gradient descent, the PHA updates the policy correction term $\tilde{\psi}$ as:
\begin{equation}
\tilde{\psi}_t = \tilde{\psi}_{t-1} - \eta_{step} \nabla_{\tilde{\psi}} D(t),
\end{equation}

where $\eta_{\mathrm{step}}$ is the gradient step size. The same update rule applies to $\widetilde{\eta}$, and the resulting intervention recommendations are constrained to valid ranges.

% \subsubsection*{partial pooling equilibrium.}

% In this case, the deception rate is such that $R_c(\psi,\eta) > 1$. Thus, $\widehat{\psi} < \psi_{HI}$ and the population reaches an endemic equilibrium.  This means that at equilibrium, the infected people reach a constant value regardless of the PHA strategies. By substituting for $I^{\prime}=0, A^{\prime}=0$ and $R^{\prime}=0$ in \eqref{eq_ODE}, the endemic equilibrium can be calculated as 
% \begin{align}
%         E^{\ast} &= \left(V^{\ast},E^{\ast},I^{\ast},A^{\ast},R^{\ast}\right) \nonumber \\
%     &=\left(V^{\ast},\frac{\gamma+\mu}{kp}I^{\ast},I^{\ast},\frac{1-p}{p}I^{\ast},\frac{\gamma}{p\mu}I^{\ast}\right).
% \end{align}

% In the presence of an endemic equilibrium, an individual chooses deception when relative utility of deception is greater than that of being honest. This can be expressed as being deceptive when
% \begin{align}
% U_s^{(k)}(c_k,m_k) > \lambda U_L.
% \label{eq_endemic_eq_condition1}
% \end{align}
% Responding to this scenario, the PHA can choose correction factors $\tilde{\psi}$ and  $\tilde{\eta}$ such that at convergence ( i.e., at   $t=\infty$),  $D(\infty) = \sum\limits_{\boldsymbol{c}} p(\boldsymbol{c} \mid \boldsymbol{m}(\infty)) 
% \sum\limits_{\boldsymbol{\theta}(\infty)} p(\boldsymbol{\theta}(\infty) \mid \boldsymbol{c}) 
%     \left| \xi  I^{\ast} - \widehat{H}(\infty) \right|^2 $ ,   where $\widehat{H}(\infty)$ corresponds  to the   estimated hospitalization rate  under  the maximum  tolerable deception.

\section{Experimental Results}

\paragraph{Experimental settings.}
We simulate $n=100$ Monte Carlo runs of a K=10{,}000-person SVEAIR population over 26 weeks (half-year) following the two stage loop of individuals reporting and PHA modeling as shown in Figure~\ref{fig:System_Model}.
Experiments span a factorial grid of equilibrium types, policy adaptivity, and experimental settings, evaluated on epidemiological, behavioral, and strategic metrics.
Details about simulation rules and evaluation procedures are provided in Appendix~\ref{supp:settings}. %We now summarize the key results of our work.

\paragraph{Interactive feedback between population and PHA drives effective epidemic control.}

\begin{figure}[h]
    \centering
    \includegraphics[width=\columnwidth]{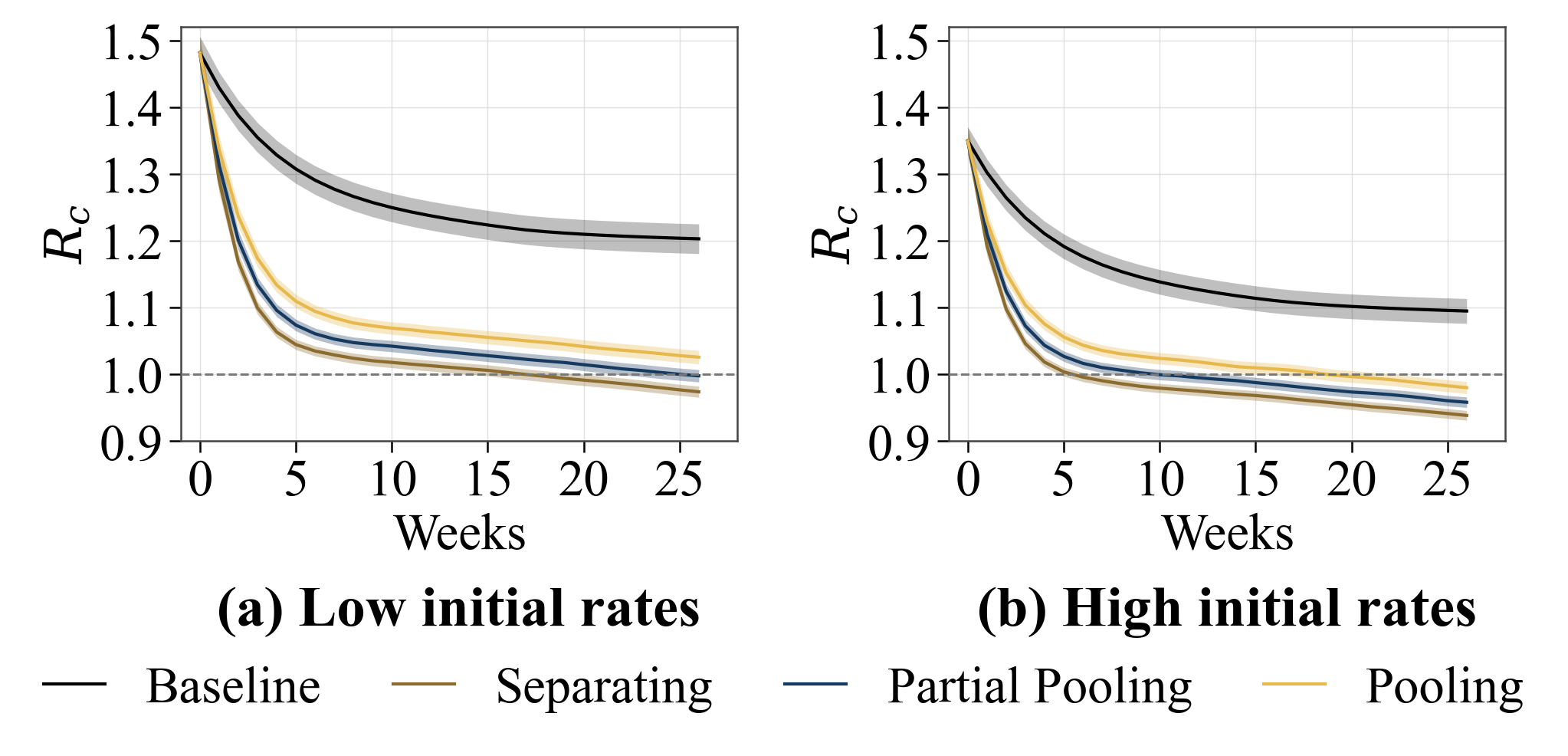}
    \caption{Controlled reproduction number ($R_c$) across equilibria vs no-interaction baseline with low (a) and high (b) initial behavioral rates.}
    \label{fig:rc} 
    \vspace{-2mm}
\end{figure}

To reflect the real-world settings, we consider two regimes: (1) low baseline behavioral rates ($\psi=0.05$, $\eta=0.10$), representing adverse conditions in which protective behaviors are largely neglected, and (2) high baseline behavioral rates ($\psi=0.15$, $\eta=0.15$), corresponding to more favorable conditions.
Figure~\ref{fig:rc} shows the evolution of $R_c$ across equilibria compared against a no-interaction baseline. 
Under low baseline rates (Figure~\ref{fig:rc}a), the separating and partial pooling equilibria achieve epidemic control ($R_c <1$) at week~17 and~24. Pooling fails to reach the epidemic threshold within the simulation horizon, indicating sustained transmission due to severely degraded information quality.
When initial behaviors are better (Figure~\ref{fig:rc}b), all three equilibria achieve control with different speeds: separating (week~5), partial pooling (week~9), pooling (week~18). 
In both scenarios, the no-interaction baseline remains $R_c > 1$, demonstrating that adaptive, behavior-informed feedback between the population and the PHA is essential for maintaining effective epidemic control.

\paragraph{Even imperfect behavioral reports provide informative signals for the PHA.}

\begin{figure}[h]
    \centering
    \includegraphics[width=\linewidth]{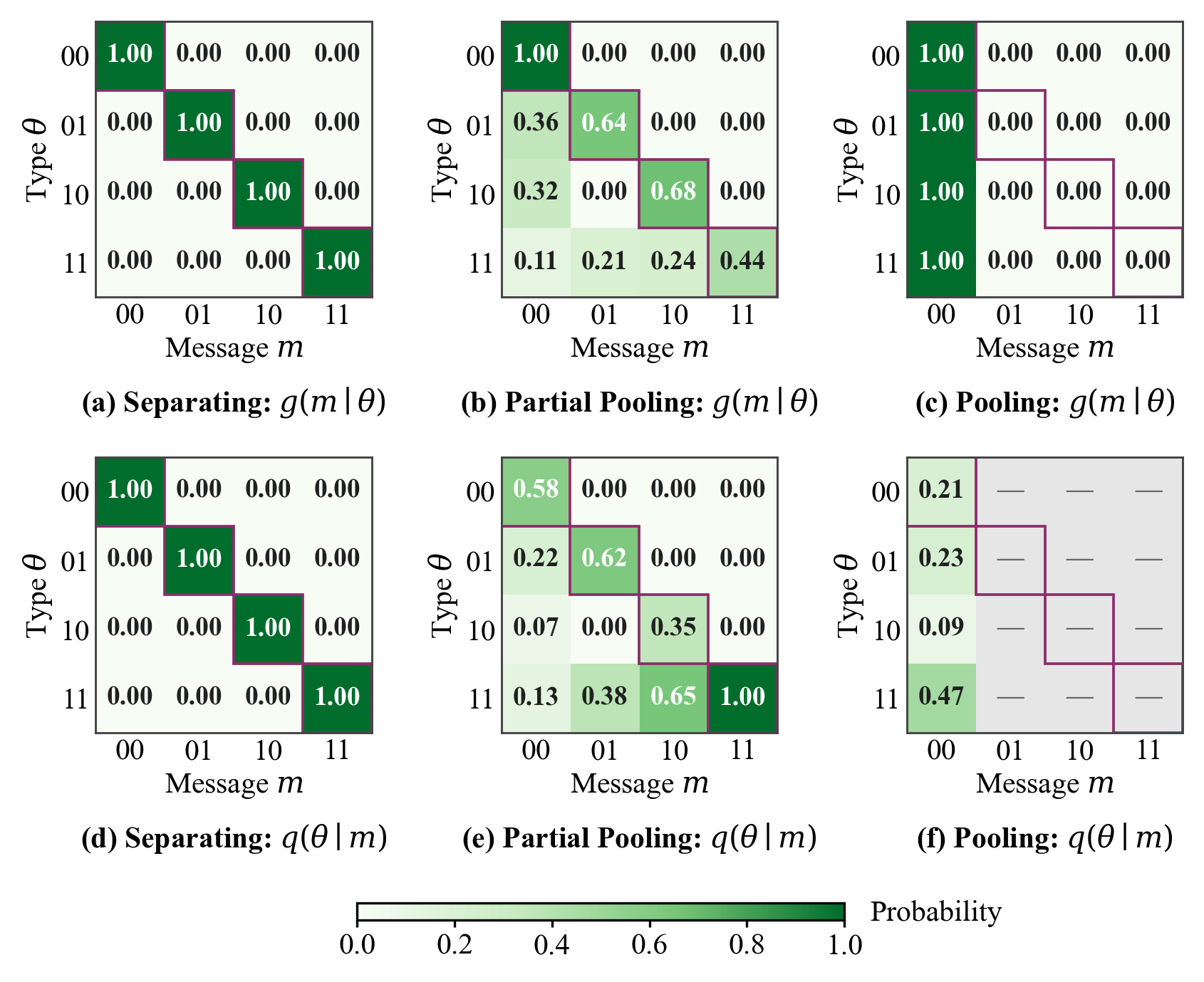}
    \caption{Sender strategies $g(m \mid \theta)$ and receiver beliefs $q(\theta \mid m)$ across equilibria.}
    \label{fig:strategy-belief}
    \vspace{-2mm}
\end{figure}

Figure~\ref{fig:strategy-belief} shows the sender communication strategy $g(m\mid\theta)$ and receiver belief $q(\theta\mid m)$ at the end of the simulation. In the separating equilibrium (Figure~\ref{fig:strategy-belief}a,d), each behavioral type sends a distinct message, allowing the receiver to recover the underlying type distribution with posterior probability one. In contrast, in the pooling equilibrium (Figure~\ref{fig:strategy-belief}c,f), all types send the same message, causing the receiver belief to collapse to the prior and making the behavioral reports uninformative.
The partial pooling equilibrium (Figure~\ref{fig:strategy-belief}b,e) represents an intermediate information regime. Although strategic misreporting introduces ambiguity between behavioral types, the received signals remain partially informative, enabling the receiver to recover non-trivial posterior beliefs. These results show that even imperfect behavioral reports can provide actionable information for the PHA, while the equilibrium structure determines the extent to which such information can be extracted.

\paragraph{Signaling game equilibria govern persistent deception patterns.}

\begin{figure}[h]
    \centering
    \includegraphics[width=\linewidth]{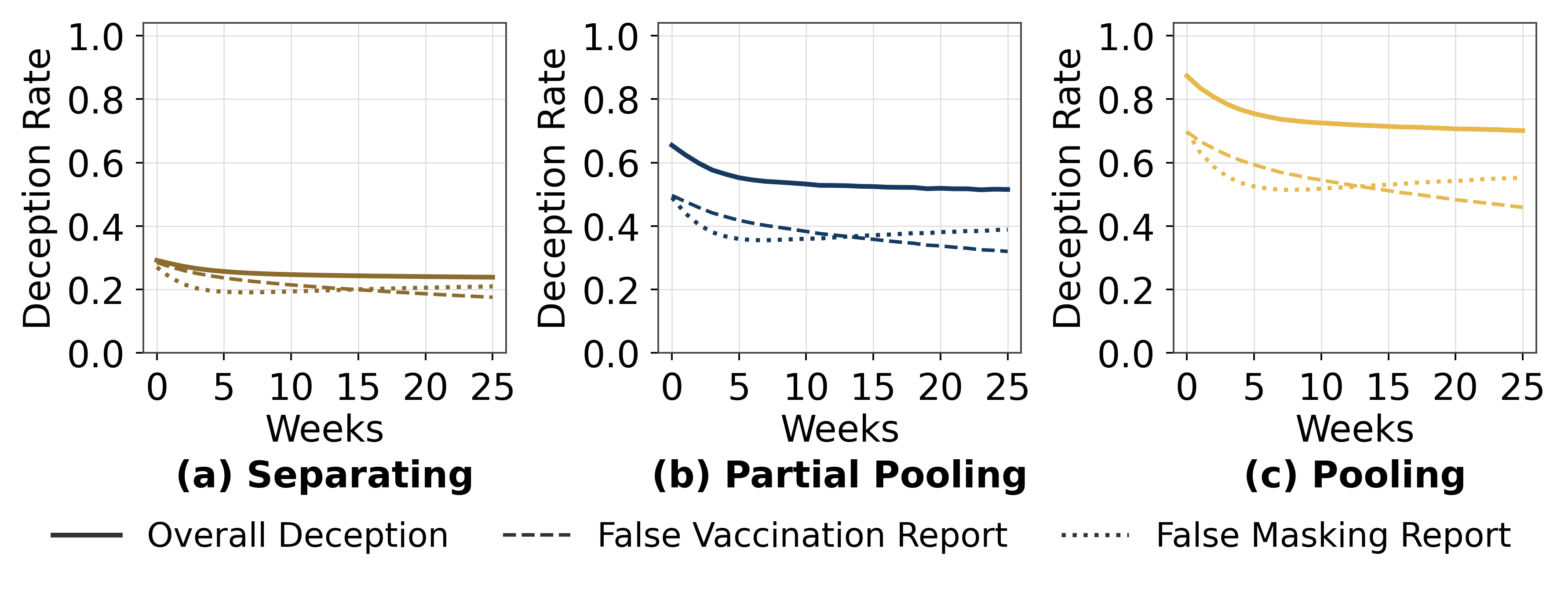}
    \caption{The deception rates across equilibria under low initial behavioral rates.}
    \label{fig:deception_low}
    \vspace{-2mm}
\end{figure}

The overall deception rate and separate deception rates for vaccination and masking are shown in Figure~\ref{fig:deception_low} and Figure~\ref{fig:deception_high}. Under low baseline rates (Figure~\ref{fig:deception_low}), pooling exhibits the highest deception: the overall deception rate starts above 0.8 and gradually declines toward about 0.65 by week 26. Partial pooling has intermediate deception, decreasing from 0.6 to 0.5, while separating maintains the lowest and most stable deception around 0.3 (reflection of non-responsive share). 
%Within each equilibrium, masking deception (dotted curves) is consistently higher than vaccination deception (dashed curves), indicating that individuals are more easily misreport masking.
The same ordering persists when baseline rates are high (Figure~\ref{fig:deception_high}). 
%Partial pooling shows a pronounced decline in masking deception, falling from above 0.5 to near 0.35 by week 26. Separating remains 0.3. 
These results highlight that equilibrium structure dominates the level of dishonesty, and that higher initial behavioral rates modestly reduces deceptive reporting.

\begin{figure*}[t]
    \centering
    \includegraphics[width=\textwidth, height=6cm]{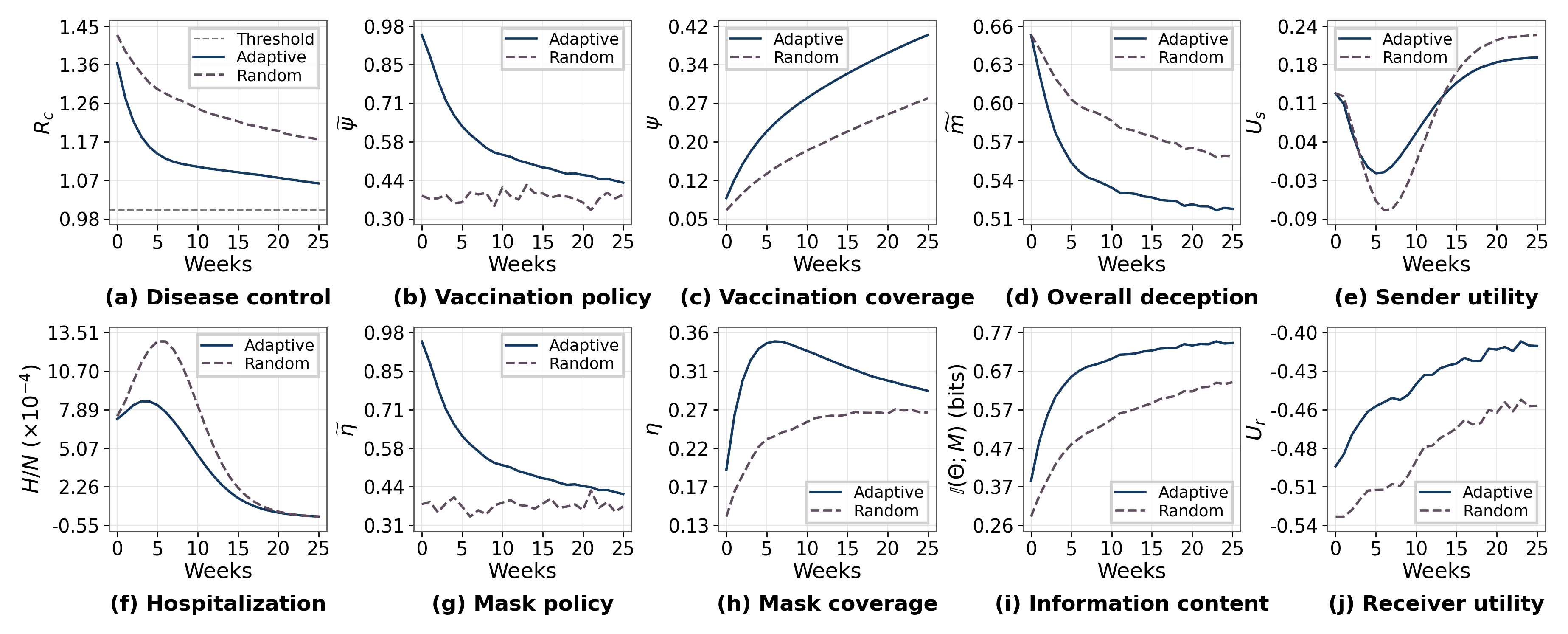}
    \caption{Comparison of adaptive and random policies under partial pooling with low initial behavioral rates, across epidemic dynamics, policy actions, signaling behavior, and utilities.}
    \label{fig:gov_policy_partial_low} %\vspace{-4mm}
\end{figure*}

\paragraph{Adaptive, signal-informed policy supports epidemic mitigation.}

To quantify the value of behavioral information, we benchmark our adaptive, signal-informed PHA response against a random policy that ignores the information of reports under partial pooling with low initial behavioral rates. 
Under the adaptive policy, $R_c$ falls below one at week 22, associated with lower peak hospitalization (0.0037 vs 0.0058) (Figure~\ref{fig:gov_policy_partial_low}a,~\ref{fig:gov_policy_partial_low}f); whereas under the random policy disease not controlled.
Adaptive feedback drives vaccination and masking intensity near their caps between weeks 3–10 (Figure~\ref{fig:gov_policy_partial_low}b,~\ref{fig:gov_policy_partial_low}g), producing rapid gains in coverage to roughly 60\% by week 10, compared to below 20\% under random control (Figure~\ref{fig:gov_policy_partial_low}c,~\ref{fig:gov_policy_partial_low}h); despite higher early deception, the overall deception rate plateaus at $0.4$ vs a stable $\sim0.5$ under random policy (Figure~\ref{fig:gov_policy}d), and more rapid growth of information content bits grow much faster, indicating more informative signals received by the PHA (Figure~\ref{fig:gov_policy_partial_low}i). Sender utility rises more under adaptive control (from 0.24 to 0.34) than under random (0.31), and receiver utility is markedly higher, peaking near 6 versus staying below 1.5 (Figure~\ref{fig:gov_policy_partial_low}e,~\ref{fig:gov_policy_partial_low}j). Similar patterns across equilibria and initial compliance are shown in Appendix~\ref{supp: policy_comparison}, where the overall quality of control follows the expected ordering separating $>$ partial pooling $>$ pooling.

\paragraph{Partial pooling quantifies the tolerance frontier between behavioral deception and epidemic control.}

\begin{figure}[t]
\centering
\includegraphics[width=\columnwidth]{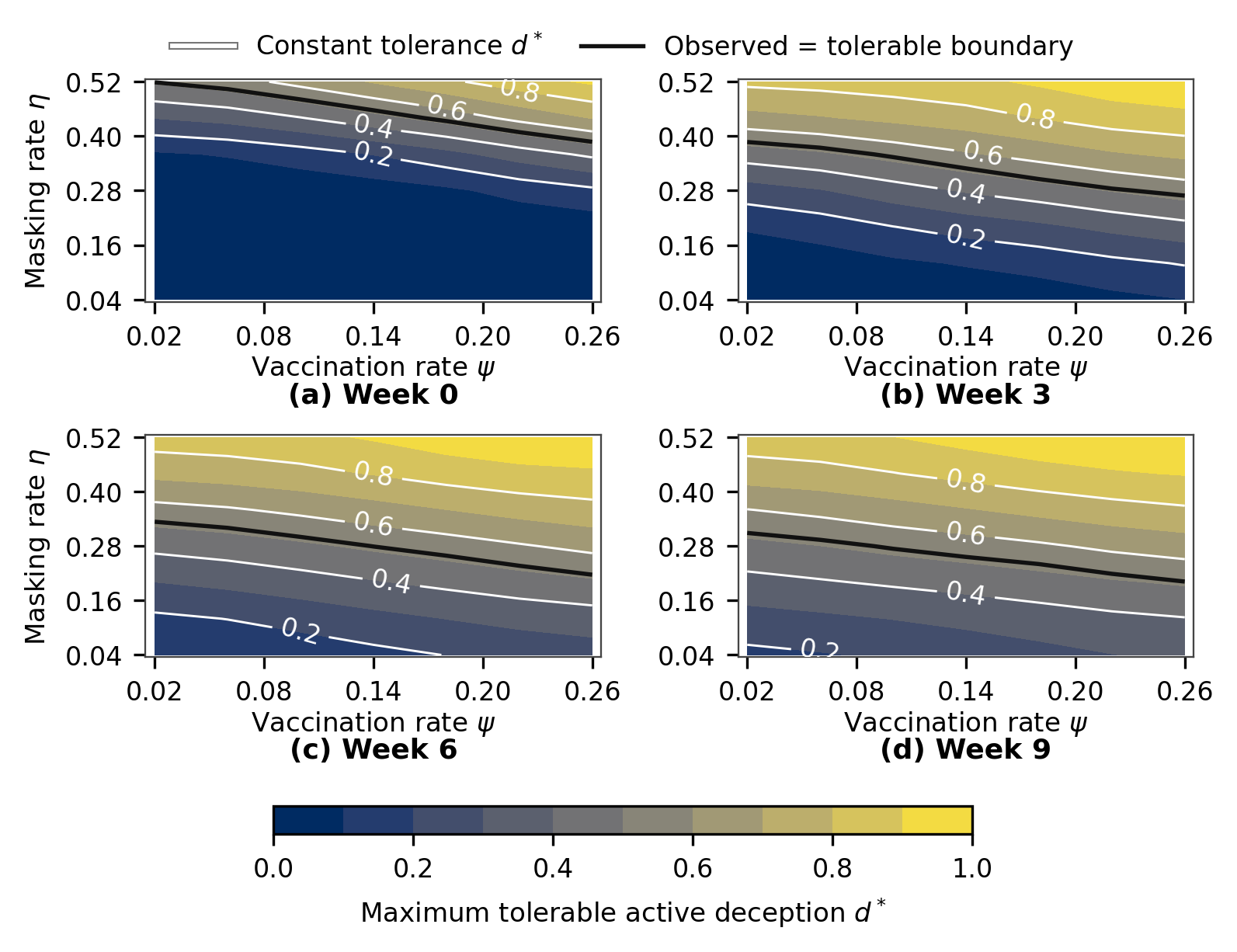} 
\caption{Maximum tolerable active deception under partial pooling across epidemic stages.}
\label{fig:tolerance_contours}% \vspace{-4mm}
\end{figure}

Figure~\ref{fig:tolerance_contours} characterizes the maximum tolerable active-deception level $d^*$ under partial pooling across combinations of vaccination rate $\psi$ and masking rate $\eta$. At Week~0, tolerance is limited when both behavioral rates are low but increases substantially as vaccination and masking coverage rises. The black curve marks combinations at which observed deception equals the tolerable level; regions above this boundary provide a positive tolerance margin. Across Weeks~3, 6, and 9, the tolerable region expands and the equality boundary shifts toward lower behavioral rates. Thus, under the modeled trajectory, the system becomes more robust to deception as behavioral protection accumulates and epidemic pressure is increasingly controlled. This pattern should not be interpreted as evidence that information quality becomes unimportant: low-coverage combinations continue to provide the smallest tolerance margins, particularly early in the outbreak. The contours also indicate complementarity between vaccination and masking. Higher adoption of either intervention can partly compensate for lower adoption of the other, although combinations with high values of both provide the greatest robustness to strategic misreporting. Overall, the figure quantifies how the behavioral protection profile and epidemic stage jointly determine the level of active deception that can be accommodated under the specified control criterion.

\paragraph{Real-world evidence supports strategic behavioral distortion modeling}

\begin{figure}[t]
\centering
\includegraphics[width=\columnwidth]{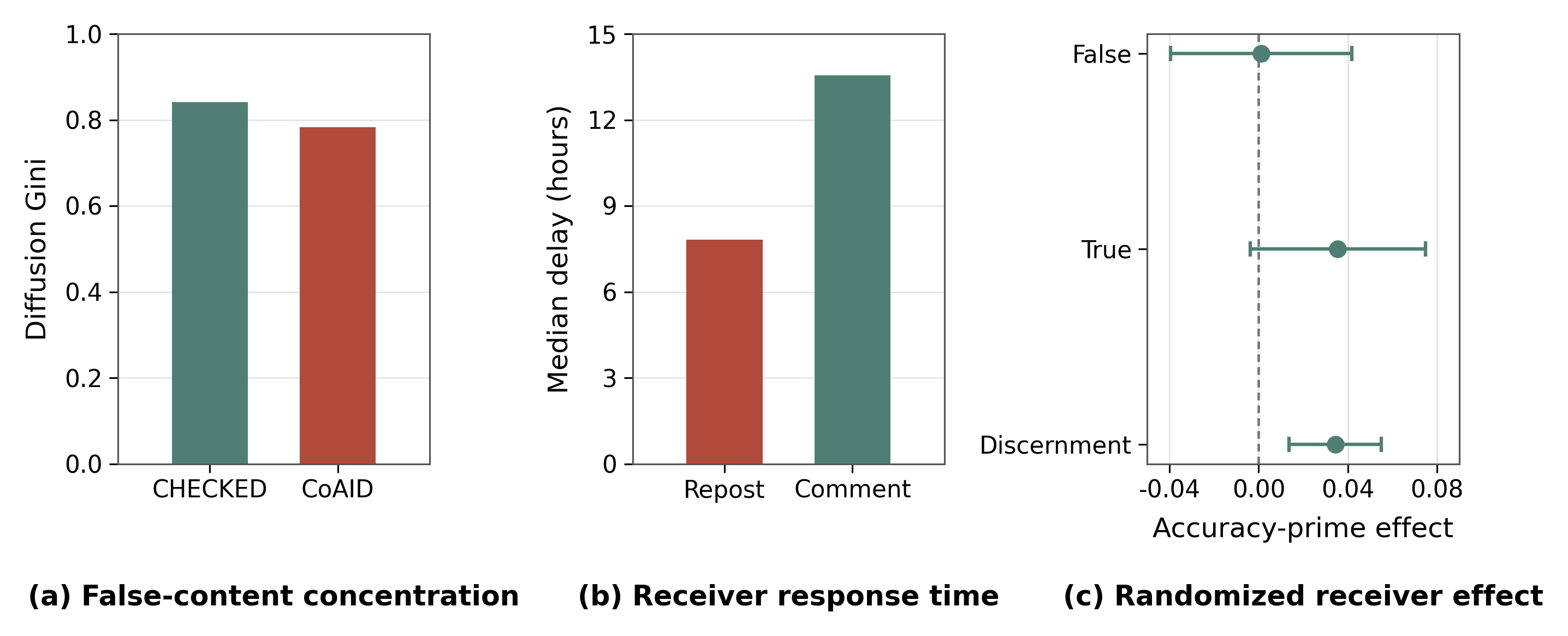} 
\caption{Tolerant deception rate.}
\label{fig:real_world}% \vspace{-4mm}
\end{figure}

To examine whether real-world behavioral distortions exhibit the
strategic structure assumed by our framework, we conduct validation on
COVID-19 misinformation, reporting, and epidemic trajectory datasets
(see details in Appendix~\ref{data}). 
Across CHECKED and CoAID, verified-false
content diffusion is highly concentrated, with content-level Gini
coefficients of 0.842 and 0.783, respectively. These results suggest
that false-information exposure follows structured patterns rather than
uniform random noise, while not implying specific sender intent. A
randomized accuracy-nudge experiment further shows that receiver
decisions are responsive to information interventions, improving
truth-versus-falsehood discernment by 0.034 on a 0--1 scale
($p=0.001$). %Together, these findings provide empirical support for modeling both structured message diffusion and adaptive receiver responses.

We further evaluate the integrated signaling--SVEAIR pipeline using CDC
vaccination and HHS hospitalization trajectories with controlled
semi-synthetic reporting distortions. The framework achieves the lowest
vaccination-state and $R_c$ errors under risk-dependent, persistent, and coordinated distortions, reducing vaccination-state error by approximately 72\%, 34\%, and 21\% relative to the strongest alternatives in these settings. Under mean-zero random errors, it performs comparably to simpler approaches, whereas temporally drifting bias favors raw-report methods. These results show that the proposed framework provides a conditional advantage when behavioral distortions exhibit structured patterns rather than arbitrary noise.

% \begin{table*}[h]
% \centering
% \caption{Stress-test outcomes across equilibria. Disease control score shows
% changes relative to baseline, while peak hospitalization reports ratios
% relative to baseline.}
% \label{tab:stress}
% \setlength{\tabcolsep}{5pt}
% \begin{tabular}{lcccccc}
% \toprule
% & \multicolumn{3}{c}{Disease Control Score ($\Delta$ vs. baseline)} & \multicolumn{3}{c}{Peak Hospitalization (ratio vs. baseline)} \\
% \cmidrule(lr){2-4}\cmidrule(lr){5-7}
% Factor & Separating & Partial Pooling & Pooling & Separating & Partial Pooling & Pooling \\
% \midrule
% Infection--hospitalization ratio & 0.065 & 0.046 & 0.028 & 1.993 & 1.985 & 1.975 \\
% Incentives & -0.085 & -0.065 & -0.042 & 1.021 & 1.029 & 1.035 \\
% Non-responsive share & -0.263 & -0.188 & -0.117 & 1.120 & 1.144 & 1.163 \\
% Vaccine efficacy & 0.325 & 0.320 & 0.243 & 0.989 & 0.981 & 0.972 \\
% \bottomrule
% \end{tabular}
% \end{table*}

\paragraph{Robustness of the SVEAIR–signaling-game pipeline to stress factors across equilibria.}

\begin{table}[h]
\centering
\caption{Stress-test outcomes across equilibria. Disease control score shows
changes relative to baseline, while peak hospitalization reports ratios
relative to baseline.}
\label{tab:stress}

\small
\setlength{\tabcolsep}{2.5pt}

\begin{tabular}{@{}lcccccc@{}}
\toprule
& \multicolumn{3}{c}{Disease Control Score ($\Delta$)}
& \multicolumn{3}{c}{Peak Hospitalization (ratio)} \\
\cmidrule(lr){2-4}
\cmidrule(lr){5-7}

Factor 
& Sep.
& Partial
& Pooling
& Sep.
& Partial
& Pooling \\

\midrule

IHR
& 0.065
& 0.046
& 0.028
& 1.993
& 1.985
& 1.975 \\

Incentives
& -0.085
& -0.065
& -0.042
& 1.021
& 1.029
& 1.035 \\

NRS
& -0.263
& -0.188
& -0.117
& 1.120
& 1.144
& 1.163 \\

Vaccine efficacy
& 0.325
& 0.320
& 0.243
& 0.989
& 0.981
& 0.972 \\

\bottomrule
\end{tabular}
\vspace{1mm} \parbox{\columnwidth}{ \footnotesize \textit{Note:} Disease control scores are reported as changes relative to the baseline setting, where positive values indicate improved control. Peak hospitalization values are ratios relative to the baseline peak hospitalization. IHR denotes the infection--hospitalization ratio and NRS denotes the non-responsive share. Sep., Partial, and Pooling denote the separating, partial pooling, and pooling equilibria, respectively. }

\end{table}

We perform a series of stress tests to validate our simulation framework under key factors, including infection–hospitalization ratios (IHRs), incentives for vaccination and masking, the non-responsive share (NRS) to surveys, and vaccine efficacy. Table~\ref{tab:stress} summarizes how robust epidemic control is to these perturbations under all equilibria. Differences in disease control scores and ratios of peak hospitalization are reported relative to the baseline setting defined in Table~\ref{tab:param-summary}. As expected, increasing the IHR approximately doubles peak hospitalization across all equilibria (1.975--1.993$\times$ the baseline peak), while producing modest improvements in disease control scores. Increasing vaccine efficacy provides the largest improvement in epidemic control, increasing the disease control score by 0.325, 0.320, and 0.243 under separating, partial pooling, and pooling equilibria, respectively, while slightly reducing peak hospitalization. By contrast, changes in incentives have relatively limited effects on epidemic outcomes, suggesting that moderate variations in behavioral incentives do not substantially alter the long-term control dynamics. Increasing the non-responsive share has a more pronounced negative impact, reducing disease control scores by 0.263, 0.188, and 0.117 across the three equilibria, highlighting the importance of obtaining sufficient behavioral participation for reliable inference. %Overall, the framework remains robust across stress factors, while the results indicate that clinical severity (IHR), intervention effectiveness (vaccine efficacy), and behavioral observability (NRS) influence epidemic control through distinct mechanisms.

\section{Conclusion}
This paper integrates a dynamic SVEAIR epidemic model with a game-theoretic signaling framework to study how strategic misreporting of vaccination and masking affects epidemic control and policy response. %By characterizing separating, partial pooling, and pooling equilibria, we 
Our results reveal main insights. Equilibrium structure largely determines both the informational content of self-reported behaviors and the achievable quality of control: separating $>$ partial pooling $>$ pooling. %Separating equilibria preserve truthful message content, allowing the PHA to track behavior accurately and achieve $R_c<1$ fastest; pooling equilibria collapse signals to the prior, forcing policies to rely on coarse, less effective interventions.
Even with nontrivial deception, infections can still be brought under control governed by signal-informed policy. %Imperfect signals remain informative: partial pooling retains enough structure so that the PHA can infer population behavior and design policies that keep $R_c$ below one.

%(3) adaptive, behavior-informed policies substantially outperform random policies that ignore signal content. Under partial pooling with low initial behavioral rates, the adaptive controller drives $R_c$ below one, accelerates coverage gains for vaccination and masking, reduces peak hospitalization, and improves both sender and receiver utilities, while the random policy remains supercritical for much of the horizon. Stress-test experiments further show that, conditional on equilibrium structure, epidemic severity is far more sensitive to clinical severity (IHR) and vaccine performance than to moderate changes in incentives or non-response, suggesting that signaling-aware policies are robust to a broad range of behavioral perturbations.

In future work, we plan to extend our framework to state-dependent deception, and richer message spaces. %that reflect modern non-pharmaceutical intervention (NPI) channels, and more detailed economic or network-based models. 
Such extensions would allow us to study how misreporting itself spreads over contact networks, and how PHAs can design feedback rules that remain effective in the face of decentralized, adaptive behavioral responses.

% \section{Acknowledgments}

% Identification of funding sources and other support, and thanks to
% individuals and groups that assisted in the research and the
% preparation of the work should be included in an acknowledgment
% section, which is placed just before the reference section in your
% document.

% This section has a special environment:
% \begin{verbatim}
%   \begin{acks}
%   ...
%   \end{acks}
% \end{verbatim}
% so that the information contained therein can be more easily collected
% during the article metadata extraction phase, and to ensure
% consistency in the spelling of the section heading.

% Authors should not prepare this section as a numbered or unnumbered {\verb|\section|}; please use the ``{\verb|acks|}'' environment.

% %%
% %% The acknowledgments section is defined using the "acks" environment
% %% (and NOT an unnumbered section). This ensures the proper
% %% identification of the section in the article metadata, and the
% %% consistent spelling of the heading.
% % \begin{acks}
% % To Robert, for the bagels and explaining CMYK and color spaces.
% % \end{acks}

\section*{Limitations and Ethical Considerations}

% Our framework makes several simplifying assumptions that can be relaxed in future work. First, the SVEAIR dynamics and the signaling layer are calibrated at the population level and do not capture individual-level heterogeneity, which may reshape epidemic trajectories, behavioral responses, and communication patterns. Second, we fix the non-responsive share and treat non-response as maximally deceptive. In practice, non-response may also arise from logistical barriers, limited access, or other non-strategic factors. Finally, we consider a single public health authority (PHA) with a fixed adaptive policy rule. Richer policy classes, including multi-agent PHAs and more diverse intervention mechanisms, may change the relative performance of equilibria and policies.
Our framework makes several simplifying assumptions that can be relaxed in future work. First, the SVEAIR dynamics and signaling layer are modeled at the population level and do not capture individual heterogeneity in behavior or transmission. Second, we fix the non-responsive share, whereas practical non-response may also arise from logistical or accessibility factors. Finally, we consider a single PHA with a fixed adaptive policy rule; richer policy classes and multi-agent decision settings may alter equilibrium outcomes.

Ethical considerations are important when applying strategic behavioral
models to public-health decision making. Our framework uses no personally
identifiable information and is designed for population-level analysis,
not individual profiling or automated decisions. However, such models
could be misused for manipulative interventions or policies that
disproportionately affect populations. Responsible deployment requires
transparent assumptions, privacy-preserving practices, diverse validation,
and human oversight.

\section*{Generative AI Usage}

During the preparation of this manuscript, the authors used generative AI tools for limited assistance with writing refinement, grammar checking, and improving presentation clarity. All research ideas, mathematical formulations, experimental designs, analyses, and conclusions were developed and verified by the authors. The authors take full responsibility for the content of this paper.

\clearpage
\bibliographystyle{ACM-Reference-Format}
\bibliography{signaling_games}

\newpage
\appendix
\setcounter{table}{0}
\renewcommand{\thetable}{A\arabic{table}}

\setcounter{figure}{0}
\renewcommand{\thefigure}{A\arabic{figure}}

\section{Epidemiological Model}
\label{supp:epi_model}
The epidemiological model studied in this paper is given by:

\begin{align}
S^{\prime}(t) &= \Lambda - \lambda(t)(1-\eta)S(t) - (\psi + \mu)S(t), \nonumber \\
V^{\prime}(t) &= \psi S(t) - \lambda(t)(1-\delta)V(t) - \mu V(t), \nonumber \\
E^{\prime}(t) &= \lambda(t)\left((1-\eta)S(t) + (1-\delta)V(t)\right)- (k+\mu)E(t), \nonumber \\
A^{\prime}(t) &= k(1-p)E(t) - (\gamma+\mu)A(t), \nonumber \\
I^{\prime}(t) &= kpE(t) - (\gamma+\mu)I(t), \nonumber \\
R^{\prime}(t) &= \gamma (A(t) + I(t)) - \mu R(t), \label{eq_ODE}
\end{align}
where $X^{\prime}(t) = \frac{d X(t)}{d t}$. We define the vaccination rate as $\psi\ (0\leq \psi \leq 1)$ and the masking rate as $\eta\ (0\leq \eta \leq 1)$, indicating that susceptible individuals reduce the contact rate by a fraction of $\eta$. 
Vaccine efficiency is described using the factor $\delta$, which quantifies the rate of reduction in infection rate for vaccinated individuals. 
Disease-free individuals enter the susceptible class through birth or immigration at a constant rate $\Lambda$ and their natural death rate is denoted by $\mu$. It is assumed that susceptible individuals become infected at a rate of $\lambda(t) = \frac{\beta_0(t)(I(t)+bA(t))}{K(t)}$. Here, $\beta_0$ is the transmission rate, $b$ is the rate of relative infectiousness of asymptomatic cases compared to symptomatic cases, and $K(t) = S(t) + V(t) + E(t) + I(t) + A(t)+R(t)$. We define $1/k$ as the duration of the latent period, which is defined as the time between when a person is exposed to the disease and when they become infectious. We also assume that a proportion $p$ of infected individuals become symptomatic.  Both symptomatic and asymptomatic individuals are assumed to recover at a rate of $\gamma$. 

Likewise, the ODE model estimated by the receiver (i.e., the PHA) is given by:

\begin{align}
\hat{S}'(t) &= \Lambda - \hat{\lambda}(t) (1 - \hat{\eta}) \hat{S}(t) - (\hat{\psi} + \mu) \hat{S}(t), \nonumber \\
\hat{V}'(t) &= \hat{\psi} \hat{S}(t) - \hat{\lambda}(t) (1 - \delta) \hat{V}(t) - \mu \hat{V}(t), \nonumber \\
\hat{E}'(t) &= \hat{\lambda}(t) \left((1 - \hat{\eta}) \hat{S}(t) + (1 - \delta) \hat{V}(t)\right) - (k + \mu) \hat{E}(t), \nonumber \\
\hat{A}'(t) &= k(1 - p) \hat{E}(t) - (\hat{\gamma} + \mu) \hat{A}(t), \nonumber \\
\hat{I}'(t) &= kp \hat{E}(t) - (\hat{\gamma} + \mu) \hat{I}(t), \nonumber \\
\hat{R}'(t) &= \hat{\gamma} (\hat{A}(t) + \hat{I}(t)) - \mu \hat{R}(t).
\label{eq_ODE_est}
\end{align}

%The ODE model as described by Eq.~\ref{eq_ODE} is illustrated by the logical flow diagram of Fig.~\ref{fig:sveair}. Next, we discuss in detail the specifics of signaling game. For notational convenience, we may omit the time index $t$; however, time dependency is understood to be implicit in the formulation.

\section{Behavioral Types}
Table~\ref{tab:truth-table-9types}
denote 9 behavioral types of individuals capturing their true status and how they report their vaccination and/or masking status. To interpret the `True Status' column, note that 0 indicates compliance and 1 indicates non-compliance. The first number in the True Status column is meant for vaccination status and the second number denotes masking status. 
In row 2, for instance, the true status is 01, i.e., the individual is vaccinated but not masking.
However, this individual reports 00 (i.e., they are honest about their vaccination status but deceptive about their masking status). The number of possible types is not a true cross-product because individuals who are masked or who are vaccinated are assumed to have no incentive to lie about their status. Thus there is a `Deceptive' value in a row only if the corresponding entry in True Status is a 1 (i.e., non-compliance).

\begin{table}[t]
\centering
\caption{Nine behavioral types defined in terms of how people with different compliance combinations report truthfully or deceive about their vaccination and/or masking status. (See A.2 for explanation.)}
\label{tab:truth-table-9types}
\begin{tabular}{l|l|l|l}
\hline
\textbf{Type} & \textbf{True status} & \textbf{Vaccine report} & \textbf{Mask report} \\
\hline
$T_1$ & 00 & 0 (Honest) & 0 (Honest) \\
$T_2$ & 01 & 0 (Honest) & 0 (Deceptive) \\
$T_3$ & 01 & 0 (Honest) & 1 (Honest) \\
$T_4$ & 10 & 0 (Deceptive) & 0 (Honest) \\
$T_5$ & 10 & 1 (Honest) & 0 (Honest) \\
$T_6$ & 11 & 0 (Deceptive) & 0 (Deceptive) \\
$T_7$ & 11 & 0 (Deceptive) & 1 (Honest) \\
$T_8$ & 11 & 1 (Honest) & 0 (Deceptive) \\
$T_9$ & 11 & 1 (Honest) & 1 (Honest) \\
\hline
\end{tabular}
\end{table}

\section{Equilibrium Characterization}
\label{supp:equilibrium}
\subsection{Pooling equilibrium}
In pooling, all types send the same message $m^* = 00$ (claim full compliance).
Here, the sender Strategy can be written as:
\begin{equation}
    g^*(m\mid \theta) = 
\begin{cases}
1 & \text{if } m = 00 \\
0 & \text{otherwise}
\end{cases}
\quad \forall \theta
\end{equation}
Receiver belief  can be written as:
\begin{equation}
    p^*(\theta \mid m = 00) = q(\theta) \quad \text{(prior probability)}
\end{equation}
The PHA learns nothing from messages and must rely on rior population compliance rates and observable hospitalizations $H(t)$.
Pooling exists when deviating to truth-telling is not profitable. For a non-compliant individual ($\theta \neq 00$):
$
U_s(m = 00 \mid \theta) \geq U_s(m = \theta \mid \theta)$.
This holds when incentives dominate semantic accuracy. 
$
I_v + I_m \succ \frac{1}{|M|} \cdot \frac{1}{K}$. 

\[
e^{-a R_c^{\text{pooling}}} (I_v + I_m) > \lambda_1 \cdot K^{-1}
\]
\emph{Receiver's Best Response:}
Under pooling, the PHA must estimate $\hat{\omega}, \hat{\epsilon}$ from hospitalizations alone, i.e.,
\[
\hat{\psi}(t) = \omega_0 + \psi_{\Delta}(t)
\]
\[
\hat{\eta}(t) = \eta + \eta_{\Delta}(t),
\]
Where corrections $\psi{\Delta}, \eta{\Delta}$ solve:
\[
\min_{\psi_{\Delta}, \eta_{\Delta}} \left| H(t) - \xi\hat{I}(t; \psi + \psi_{\Delta}, \eta + \eta_{\Delta}) \right|^2
\]
subject to distortion constraint:
$
D(t) \leq D^*$.
\emph{Deception Rate}
can be written as:
$\tilde{m}^* = \frac{1}{2} \sum_{k \in K_c} \| m_k - \theta_k \|_H + |K_s|
$, where $\| m_k - \theta_k \|_H$ is Hamming distance. In pure pooling with $m = 00$ for all, we obtain
$
\tilde{m}^* = \frac{1}{2} |K_c| \cdot \mathbb{E}[\| 00 - c \|_H] + |K_s|
$.
If population has baseline rates $(\psi^\ast, \eta^\ast)$ (round truth), then we get
\begin{equation}
\begin{aligned}
\tilde{m}^* &= |K_c|(1 - \psi^\ast)(1 - \eta^\ast) + \frac{1}{2}|K_c|\psi^\ast(1 - \eta^\ast) \\ &+\frac{1}{2}|K_c|(1-\psi^\ast) \eta^\ast+  |K_s|.
\end{aligned}
\end{equation}

\subsection{Partial pooling equilibrium}
\label{supp:equilibrium-pooling}
Here, some users separate and others pool, with mixing probability defined as $\alpha$. We consider a simplistic model of sender strategy as written below:
\begin{equation}
    g^*(m\mid \theta) =
\begin{cases}
1 & \text{if } \theta = 00 \text{ and } m = 00 \\
\alpha & \text{if } \theta \in \{01,10\} \text{ and } m = 00 \\
1-\alpha & \text{if } \theta \in \{01,10\} \text{ and } m = \theta \\
1 & \text{if } \theta = 11 \text{ and } m = 00 \\
0 & \text{otherwise}
\end{cases}
\end{equation}
Let:
$\pi_{00}$ represents fraction truly compliant (vaccinated + masking),
 $\pi_{01}$ represents fraction vaccinated but not masking
$\pi_{10}$  represents fraction masking but not vaccinated, and 
 $\pi_{11}$  represents  fraction non-compliant on both, with
$
\pi_{00} + \pi_{01} + \pi_{10} + \pi_{11} = 1.
$ Given mixing probability $\alpha$, the proportion sending each message becomes,
\begin{equation}
\begin{aligned}
P(m = 00) &= \pi_{00} + \alpha(\pi_{01} + \pi_{10}) + \pi_{11}\\
P(m = 01) &= (1 - \alpha)\pi_{01}\\
P(m = 10) &= (1 - \alpha)\pi_{10}\\
P(m = 11) &= 0 \quad \text{(no one admits full non-compliance)}    
\end{aligned}
\end{equation}
Using Bayes' rule, when PHA receives $m = 00$, and $\theta_{ij}$ representing the estimate of type $``ij"$,
\begin{equation}\begin{aligned}
q(\theta_{00} \mid m = 00) &= \frac{\pi_{00}}{\pi_{00} + \alpha(\pi_{01} + \pi_{10}) + \pi_{11}}\\
q(\theta_{01} \mid m = 00) &= \frac{\alpha \pi_{01}}{\pi_{00} + \alpha(\pi_{01} + \pi_{10}) + \pi_{11}}\\
q(\theta_{10} \mid m = 00) &= \frac{\alpha \pi_{10}}{\pi_{00} + \alpha(\pi_{01} + \pi_{10}) + \pi_{11}}\\
q(\theta_{11} \mid m = 00) &= \frac{\pi_{11}}{\pi_{00} + \alpha(\pi_{01} + \pi_{10}) + \pi_{11}}\end{aligned}
\end{equation}
The PHA estimates vaccination rate as:
\begin{equation}
\begin{aligned}
\hat{\psi}(\alpha) &= \pi_{00} + \pi_{01} + \alpha\pi_{10}+ \pi_{11}
\end{aligned}
\end{equation}
\begin{equation}
\hat{\eta}(\alpha) = \pi_{00} + \pi_{10} + \alpha\pi_{01}+ \pi_{11}
\end{equation}
Similarly for masking:
\begin{equation}
\hat{\epsilon}(\alpha) = P(m = 00) \cdot \frac{\pi_{00} + \alpha \pi_{10}}{\pi_{00} + \alpha(\pi_{01} + \pi_{10}) + \pi_{11}} + (1 - \alpha)\pi_{10}
\end{equation}
\textit{Key observation:} $\hat{\psi}(\alpha)$ and $\hat{\eta}(\alpha)$ are overestimates of true compliance because of pooling. From eq.~\eqref{eq_Rc}, $R_c(\alpha) = R_c(\hat{\psi},\hat{\eta})$. As $\alpha$ increases (indicating more pooling), $\hat{\psi}(\alpha)$ and $\hat{\eta}(\alpha)$ both increase because the PHA believes more individuals are vaccinated and masking, even though true compliance remains unchanged; consequently, the PHA under-responds, leading to an increase in $R_c(\alpha)$. The semantic accuracy term can be written when the actual type is $\theta_{01}$ as:
\begin{equation}
U_{L,\text{pool}}^{01}(\alpha) = \alpha \cdot q(\theta_{01} \mid m = 00) + (1 - \alpha)
\end{equation}
For type $\theta_{01}$ pooling:
\begin{equation}
U^{01}_{\text{pool}}(\alpha) = e^{-a R_c(\alpha)} I_m + \lambda_1 \cdot \big[ \alpha \cdot q(\theta_{01} \mid m = 00) + (1 - \alpha) \big],
\end{equation}
and for type $\theta_{01}$ separating:
\begin{equation}
U^{01}_{L,\text{sep}} = 0 + \lambda_1 \cdot 1 = \lambda_1
\end{equation}
For optimal mixing,
we need
$U^{01}_{\text{pool}}(\alpha^*) = U^{01}_{\text{sep}}. $ This leads to:
\begin{equation}
e^{-a R_c(\alpha^*)} I_m + \lambda_1 \big[ \alpha^* \cdot q(\theta_{01} \mid m = 00) + (1 - \alpha^*) \big] = \lambda_1
\end{equation}
Simplifying this, we obtain
\begin{equation}
\alpha^* = \frac{e^{-a R_c(\alpha^*)} I_m}{\lambda_1 \big[ 1 - q(\theta_{01} \mid m = 00) \big]}.
\end{equation}
Substituting Bayes' rule, we obtain:
\begin{equation}
e^{-a R_c(\alpha^*)} I_m = \lambda_1 \alpha^* \left[ 1 - \frac{\alpha^* \pi_{01}}{\pi_{00} + \alpha^*(\pi_{01} + \pi_{10}) + \pi_{11}} \right]
\end{equation}
This is an implicit equation in $\alpha^*$ because $R_c(\alpha^*)$ depends on $\alpha^*$ and the right-hand side contains $\alpha^*$ in a nonlinear way. Therefore, it generally requires a numerical solution. For small deviations from separation, we use a linear approximation:
\begin{equation}
R_c(\alpha) \approx R_c(0) + \beta \alpha
\end{equation}
where
\begin{equation}
\beta = \left. \frac{d R_c}{d \alpha} \right|_{\alpha = 0} > 0
\end{equation}
Similarly, the conditional probability can be approximated as:
\begin{equation}
q(\theta_{01} \mid m = 00) \approx \frac{\alpha \pi_{01}}{\pi_{00} + \pi_{11}}
\end{equation}
Then, using the linear approximation for $R_c(\alpha)$ and ignoring higher-order terms, we have:
\begin{equation}
e^{-a [R_c(0) + \beta \alpha]} I_m \approx \lambda_1 \alpha \left[ 1 - \frac{\alpha \pi_{01}}{\pi_{00} + \pi_{11}} \right]
\end{equation}
For small $\alpha$, ignoring $\alpha^2$ terms, this simplifies to:
\begin{equation}
e^{-a R_c(0)} (1 - a \beta \alpha) I_m \approx \lambda_1 \alpha
\end{equation}
Rearranging terms gives:
\begin{equation}
e^{-a R_c(0)} I_m \approx \lambda_1 \alpha + a \beta e^{-a R_c(0)} I_m \cdot \alpha
\end{equation}
Finally, solving for $\alpha^*$ yields:
\begin{equation}
\alpha^* \approx \frac{e^{-a R_c(0)} I_m}{\lambda_1 + a \beta e^{-a R_c(0)} I_m}
\end{equation}

\subsubsection{Existence conditions}
For partial pooling equilibrium to exist, we require $0 < \alpha^* < 1$. To find a lower bound, consider when pure separation ($\alpha = 0$) is not an equilibrium. At $\alpha = 0$, suppose type $\theta_{01}$ deviates to pool. The utility from deviating is:
\begin{equation}
\begin{aligned}
U_{L,\text{pool}}^{01}(\alpha=0)
&= e^{-a R_c(0)} I_m + \lambda_1 \big[(1 - 0) + 0 \cdot 0 \big] \\
&= e^{-a R_c(0)} I_m + \lambda_1 .
\end{aligned}
\end{equation}
The utility from separating is:
\begin{equation}
U_{L,\text{sep}}(\alpha=0) = \lambda_1
\end{equation}
For separation to be unstable, we need:
\begin{equation}
e^{-a R_c(0)} I_m + \lambda_1 > \lambda_1
\end{equation}
This simplifies to:
\begin{equation}
e^{-a R_c(0)} I_m > 0
\end{equation}
which is always true, implying that pure separation cannot be an equilibrium when $I_m > 0$. To determine when pure pooling ($\alpha = 1$) is not an equilibrium, note that at $\alpha = 1$, everyone pools. The conditional probability is:
\begin{equation}
q(\theta_{01} \mid m = 00) = \frac{\pi_{01}}{\pi_{00} + \pi_{01} + \pi_{10} + \pi_{11}} = \pi_{01}
\end{equation}
The utility from pooling is:
\begin{equation}
\begin{aligned}
U_{L,\text{pool}}(1) &= e^{-a R_c(1)} I_m + \lambda_1 \big[(1 - 1) + 1 \cdot \pi_{01}\big]\\& = e^{-a R_c(1)} I_m + \lambda_1 \pi_{01}
\end{aligned}
\end{equation}
The utility from deviating to separate is:
\begin{equation}
U_{L,\text{sep}} = \lambda_1
\end{equation}
For pooling to be unstable (deviation profitable), we require:
\begin{equation}
e^{-a R_c(1)} I_m + \lambda_1 \pi_{01} < \lambda_1
\end{equation}
Since $(1 - \pi_{01}) \leq 1$, the condition for pooling to be unstable simplifies to:
\begin{equation}
e^{-a R_c(1)} I_m < \lambda_1
\end{equation}
From the approximation
\begin{equation}
\alpha^* \approx \frac{e^{-a R_c(0)} I_m}{\lambda_1 + a \beta e^{-a R_c(0)} I_m},
\end{equation}
for $\alpha^* < 1$, we require:
\begin{equation}
e^{-a R_c(0)} I_m < \lambda_1 + a \beta e^{-a R_c(0)} I_m.
\end{equation}
Rearranging gives:
\begin{equation}
e^{-a R_c(0)} I_m (1 - a \beta) < \lambda_1.
\end{equation}

For small $a \beta$, this implies:
\begin{equation}
\lambda_1 > e^{-a R_c(0)} I_m - a \beta e^{-a R_c(0)} I_m.
\end{equation}

Combining this lower bound on $\lambda_1$ with the upper bound from pooling instability, we obtain:
\begin{equation}
e^{-a R_c(1)} I_m < \lambda_1 < e^{-a R_c(0)} I_m + a \beta e^{-a R_c(0)} I_m.
\label{bounds_partialpooling}
\end{equation}
From \eqref{bounds_partialpooling} we infer the following conclusions.
The semantic weight $\lambda_1$ must lie in a Goldilocks zone for mixed to exist. If $\lambda_1$ is too low (below the lower bound), semantic accuracy does not matter much, incentives dominate, and everyone pools ($\alpha = 1$), resulting in a full pooling equilibrium. If $\lambda_1$ is too high (above the upper bound), semantic accuracy matters a lot, people value being understood correctly, and everyone separates ($\alpha = 0$), leading to a separating equilibrium.
If $\lambda_1$ is just right (within the window), there is a trade-off between incentives and accuracy, and mixed occurs with $0 < \alpha^* < 1$. In this case, some individuals lie while others tell the truth, creating a stable mixing equilibrium.

\subsection{Tolerance frontiers}

In partial pooling equilibrium, the key question is: \emph{to what extent can deception be tolerated while still ensuring the epidemic converges to an equilibrium with negligible infections?}

\begin{theorem}[Partial Pooling and Tolerable Deception]
\label{thm:partial}
There exists a partial pooling equilibrium in which some non-compliant individuals misreport while others report truthfully if and only if the induced deception level $\tilde{m}$ remains below a critical tolerance threshold $\tilde{m}_{\max}$ such that the PHA can select policies satisfying
\[
R_c(\psi,\eta) < 1 \quad \text{and} \quad \widehat{\psi} > \psi_{HI}.
\]
In this regime, epidemic control is preserved despite strategic misreporting.
\end{theorem}

The population strategy that maximizes deception while preserving epidemic control can be expressed as:
\begin{equation}
\begin{aligned}
g^{\ast}(\boldsymbol{m}\mid \boldsymbol{\theta})
= \argmax_{g(\boldsymbol{m}\mid\boldsymbol{\theta}),\,\tilde{\psi}} \ \tilde{m}
\quad
\text{s.t.}\quad
\psi^{\ast} + \tilde{\psi} > \psi_{HI}.
\end{aligned}
\end{equation}

The resulting equilibrium deception level is given by:
\begin{equation}
\begin{aligned}
\tilde{m}^{\ast}
=
\frac{1}{2}\sum_{k\in \mathcal{K}_c}\sum_{m_k} g^{\ast}(m_k\mid \theta_k)\lVert m_k-m_k^{(\ast)} \rVert_H
+ |\mathcal{K}_s|.
\end{aligned}
\end{equation}

This characterization defines a \emph{tolerance frontier} that quantifies how much strategic misreporting can be absorbed before epidemic control breaks down.

We further decompose the deception into vaccination and masking components.
Let $h_v(m_k,m_k^{(*)}) \in \{0,1\}$ and $h_\eta(m_k,m_k^{(*)}) \in \{0,1\}$
denote the Hamming distance restricted to the vaccination bit and the
masking bit, respectively. Then the vaccination deception $\tilde{m}^{(v)}$
and masking deception $\tilde{m}^{(\eta)}$ can be written as
\begin{align}
\tilde{m}^{(v)}
  &= \frac{1}{|\mathcal K|}
\left[\sum_{k \in \mathcal{K}_c} \sum_{m_k} g(m_k \mid \theta_k)\,
     h_v\!\bigl(m_k, m_k^{(*)}\bigr)
     + \abs{\mathcal{K}_s}\right], \label{eq:vax-deception} \\[4pt]
\tilde{m}^{(\eta)}
  &= \frac{1}{|\mathcal K|}
\left[\sum_{k \in \mathcal{K}_c} \sum_{m_k} g(m_k \mid \theta_k)\,
     h_\eta\!\bigl(m_k, m_k^{(*)}\bigr)
     + \abs{\mathcal{K}_s}\right]. \label{eq:mask-deception}
\end{align}

\section{Weekly Interactive Loop Between Population and PHA}
\label{supp:weekly-loop}

\begin{table*}[t]
\centering
\begin{minipage}{\textwidth}
\begin{algorithm}[H]
\caption{Weekly Interactive Loop Between Population and PHA}
\label{alg:weekly_loop}
\begin{algorithmic}[1]
\Require Initial states $(S,V,E,A,I,R)$, initial policies $(\psi_0,\eta_0)$, parameters $(\beta_0,k,p,\gamma,\mu,\xi)$
\For{each week $t = 1$ to $T$}
    \State \textbf{Population reporting:}
    \For{each individual $i \in \mathcal{K}$}
        \State Decide true behaviors $(v_i, m_i)$ (vaccination, masking)
        \State Generate reported message $(\hat{v}_i, \hat{m}_i)$ according to equilibrium type (truthful or deceptive)
    \EndFor
    \State Aggregate reports to obtain $\hat{\psi}_t$, $\hat{\eta}_t$, and deception rate $d_t$
    \Statex

    \State \textbf{PHA update:}
    \State Observe hospitalization $H_t = \xi I_t$
    \State Use internal SVEAIR model and current beliefs to compute $\hat{H}_t$
    \State Compute distortion $D_t = (H_t - \hat{H}_t)^2$
    \If{$D_t > D^*$}
        \State Update policy recommendations via a gradient step:
        \Statex\hspace{1.2em}$\psi_{r,t} = \psi_t + \Delta_\psi,\quad \eta_{r,t} = \eta_t + \Delta_\eta$
        \State where $(\Delta_\psi,\Delta_\eta)$ are obtained by minimizing $D_t$
    \Else
        \State Keep previous recommendations $(\psi_{r,t}, \eta_{r,t}) \gets (\psi_t, \eta_t)$
    \EndIf
    \Statex

    \State \textbf{Epidemic propagation:}
    \State Integrate SVEAIR ODEs for one week using $(\psi_{r,t}, \eta_{r,t})$
    \State Update population states $(S,V,E,A,I,R)_{t+1}$
    \State Store weekly metrics: $R_c(t)$, deception rate $d_t$, distortion $D_t$, utilities
\EndFor
\end{algorithmic}
\end{algorithm}
\end{minipage}
\end{table*}

Algorithm~\ref{alg:weekly_loop} summarizes the iterative feedback mechanism that governs the interaction between the population and the PHA in each simulation cycle. Each week, individuals report (truthfully or deceptively) their behavioral status, and the PHA updates its epidemiological model and policy recommendations based on these reports and observed hospitalization outcomes. The two-stage loop between players sketched in Figure~\ref{fig:System_Model}:

\begin{enumerate}
\item {\bf Individuals reporting.} Individuals decide whether to reveal or to misreport their current vaccination and masking status, weighing economic incentives against the costs of deception. To reflect realistic behavioral constraints, only non-compliers who forego either vaccination or masking are eligible to lie for the corresponding behavior. Moreover, individuals in the recovered (R) compartment are barred from deception because their acquired immunity leaves them little incentive to falsify their status.

\item {\bf Public Health Authority (PHA) updates.} The PHA receives the (possibly distorted) reports, updates its belief over population behavior, and issues revised policy recommendations for vaccination uptake and contact-rate reduction.
\end{enumerate}

\section{Experimental Settings}
\label{supp:settings} 

\subsection{Experimental conditions.}
We evaluate the model over a full factorial grid: three equilibrium structures (separating, partial pooling, pooling) crossed with public-health policy (adaptive vs. random), baseline behavioral-rate level (low vs. high), and incentives for vaccination and masking.
Daily outputs are aggregated into weekly steps. Epidemiological and behavioral series are smoothed with a Savitzky–Golay filter, whereas strategic variables are updated by exponential smoothing to retain short-term fluctuations.

\subsection{Baselines.} Two baselines benchmark the two-player signaling game schemes studied here:

\begin{enumerate}
\item {\bf No interaction}. Individuals never report behaviors and the PHA issues no recommendations, providing a lower-bound for epidemic control across equilibria.

\item {\bf Random policy.} PHA replaces its adaptive rule with weekly recommendations perturbed by zero-mean noise whose variance matches the empirical variance of the adaptive policy’s gradient (estimated from hospitalisations), isolating the value of behavior-aware adaptation.
\end{enumerate}

\subsection{Datasets}
\label{data}
\subsubsection{Simulations}
As described in the main text, for each scenario considered, we run Monte-Carlo simulations of a K=10{,}000–person population over a 26-week horizon with the baseline parameter set listed in Table~\ref{tab:param-summary}.

\begin{table*}[h]
\begin{minipage}{\textwidth}
\centering
\caption{Summary of model parameters and baseline values used in our experiments.}
\label{tab:param-summary}
% \small
\begin{tabular}{l|l|l|l}
\toprule
\textbf{Parameter} & \textbf{Description} & \textbf{Value} & \textbf{References} \\
\midrule
\multicolumn{4}{l}{\textit{Epidemiological model parameters}} \\
\midrule
$\mu$         & Natural death rate                         & $1/(75 \times 365)$ & ~\cite{Longini2004}\\
$\beta_0$     & Baseline transmission rate                 & $0.35$ & ~\cite{JING2021}\\
% $\psi$        & Vaccination rate                           & varies \\
$\delta$      & Vaccine efficacy                                      & $0.45$ & ~\cite{Dawood2020}\\
$1/\gamma$    & Infectious period (days)                              & $10$ & ~\cite{Walsh2020}\\
$b$           & Infectiousness of asymptomatic                        & $0.5$ & ~\cite{Longini2004}\\
% $\eta$        & Masking/contact reduction                             & varies \\
$1/k$         & Latent period (days)                                  & $5$ & ~\cite{RELUGA2011}\\
$p$           & Probability exposed to infected                       & $0.67$ &~\cite{Longini2004}\\
$\xi$         & Hospitalization ratio                                 & $0.05$ & ~\cite{FoxPNAS2022} \\
\midrule
\multicolumn{4}{l}{\textit{Signaling-game parameters}} \\
\midrule
$I_v$         & Lying incentive for vaccination                       & $1.0$ & ~NA\\
$I_m$         & Lying incentive for masking                           & $0.5$ & ~NA\\
$\lambda_1$   & Weight on semantic loss                               & $0.2$ & ~NA\\
$\alpha$      & Surprise weight in receiver utility                   & $1.0$ & ~NA\\
$a$           & Economic factor                                       & $0.5$ & ~NA\\
% \hline
% \multicolumn{3}{l}{\textit{Survey and policy parameters}} \\
% \hline
% non\_responsive\_share & Non-responsive population share              & $0.3$ \\
% $\psi_{\max}$ & Max vaccination recommendation                        & $0.95$ \\
% $\psi_{\min}$ & Min vaccination recommendation                        & $-0.05$ \\
% $\eta_{\max}$ & Max masking recommendation                            & $0.90$ \\
% $\eta_{\min}$ & Min masking recommendation                            & $-0.05$ \\
\midrule
\multicolumn{4}{l}{\textit{Initial conditions}} \\
\midrule
$\psi_{\text{init}}$ & Baseline vaccination rate                      & $0.05$  & ~\cite{Oguz2019}\\
$\eta_{\text{init}}$ & Baseline masking rate                      & $0.10$ & ~\cite{pnas2021}\\
$I_0$                & Initial infected count                        & $150$  & ~NA\\
$P(K_s)$             & Non-responsive population share               & $0.3$ & ~\cite{sciadv.adj0266}\\
\midrule
\multicolumn{4}{l}{\textit{Simulation parameters}} \\ 
\midrule
$K$           & Population size                                       & $10\,000$ & ~NA \\
$T$           & Simulation duration (weeks)                           & $26$ & ~\cite{NEJMc}\\
\bottomrule
\end{tabular}
\end{minipage}
\end{table*}

\subsubsection{Real-world data}
We evaluate the framework on real-world COVID-19 datasets.
The CHECKED dataset contains 332 verified-false COVID-19 Weibo posts and
26,501 timestamped receiver actions. The CoAID dataset contains 166 false
and 1,730 reliable COVID-19 items with associated tweets and replies. For
receiver-side accuracy evaluation, we use an intervention dataset with
25,627 headline ratings collected from 855 randomized participants.
For epidemic-state validation, we use CDC--HHS vaccination and
hospitalization trajectories covering 52 U.S. states over 27 weeks and
construct semi-synthetic reporting channels with controlled distortion
mechanisms. Additional Germany, U.S. survey, and Gujarat audit datasets
are included as supplementary validation cases.

\subsection{Experimental conditions.}
We evaluate the model over a full factorial grid: three equilibrium structures (separating, partial pooling, pooling) crossed with public-health policy (adaptive vs. random), baseline behavioral-rate level (low vs. high), and incentives for vaccination and masking.
Daily outputs are aggregated into weekly steps. Epidemiological and behavioral series are smoothed with a Savitzky–Golay filter, whereas strategic variables are updated by exponential smoothing to retain short-term fluctuations.

\subsection{Evaluation metrics.} We measure the $week_{control}$  when $R_c$ first falls $<$ 1 (converted to a disease-control score as $1-week_{control}/26$), deception rate, vaccination and masking coverage, peak hospitalization, sender utility (benefit from recommendations) and receiver distortion (cost from deviating signals). 

\subsection{Data and code.}
All code for running simulations, generating figures, and reproducing results is available at: \url{https://github.com/yiqisu/epi-signaling-games}.

\section{Supplemental experimental results}
\subsection{Comparison of adaptive and random policies across equilibria.}
\label{supp: policy_comparison}
Across all three information regimes, the adaptive policy consistently outperforms the random policy in epidemic control, driving lower $R_c$, smaller peak hospitalization, and better receiver utility. Moreover, the overall quality of control follows the expected ordering separating $>$ partial pooling $>$ pooling, with both peak hospitalization and equilibrium deception rates smallest under separating, larger under partial pooling, and largest under pooling. 

An interesting nuance is that when initial compliance is already high, sender utility under the random policy is often comparable to, or slightly higher than, under adaptive control—because in that regime the epidemic can be kept in check even with relatively mild, noisy recommendations, so the population “pays” less in terms of strong, targeted interventions while still enjoying reasonably low infection risk.

\subsection{Deception patterns}

\begin{figure}[h]
    \centering
    \includegraphics[width=\linewidth]{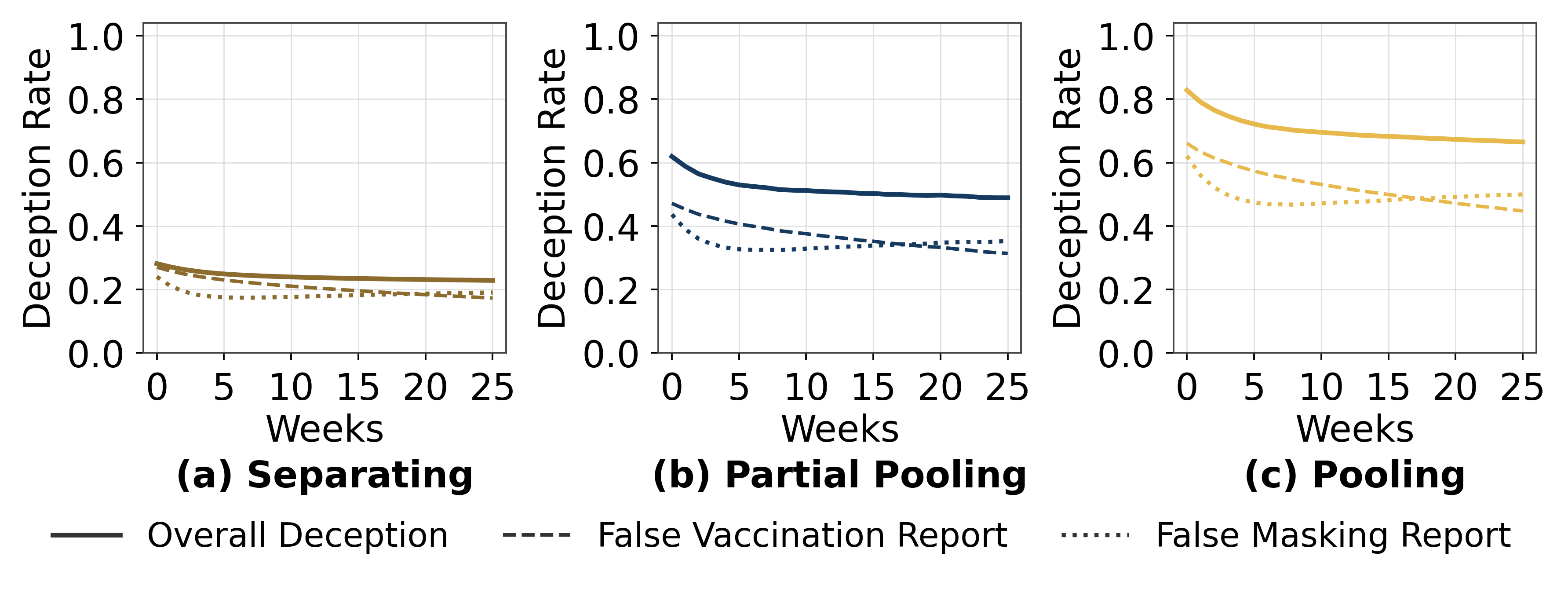}
    \caption{The deception rates across equilibria under high initial behavioral rates.}
    \label{fig:deception_high}
    % \vspace{-3mm}
\end{figure}
Figure~\ref{fig:deception_high}: The overall and separate strategy deception rates across equilibria under high initial behavioral rates.

% \subsection{Stress test}

% \begin{figure}[htbp]
%     \centering
%     \includegraphics[width=\columnwidth]{figures/stress_summary.png}
%     \caption{Sensitivity of epidemic outcomes to stress factors across equilibria.}
%     \label{fig:stress}
% \end{figure}  

% Figure~\ref{fig:stress}: Sensitivity of epidemic outcomes to stress factors across equilibria.

\begin{figure*}[h]
    \centering
    \includegraphics[width=\textwidth, height=5cm]{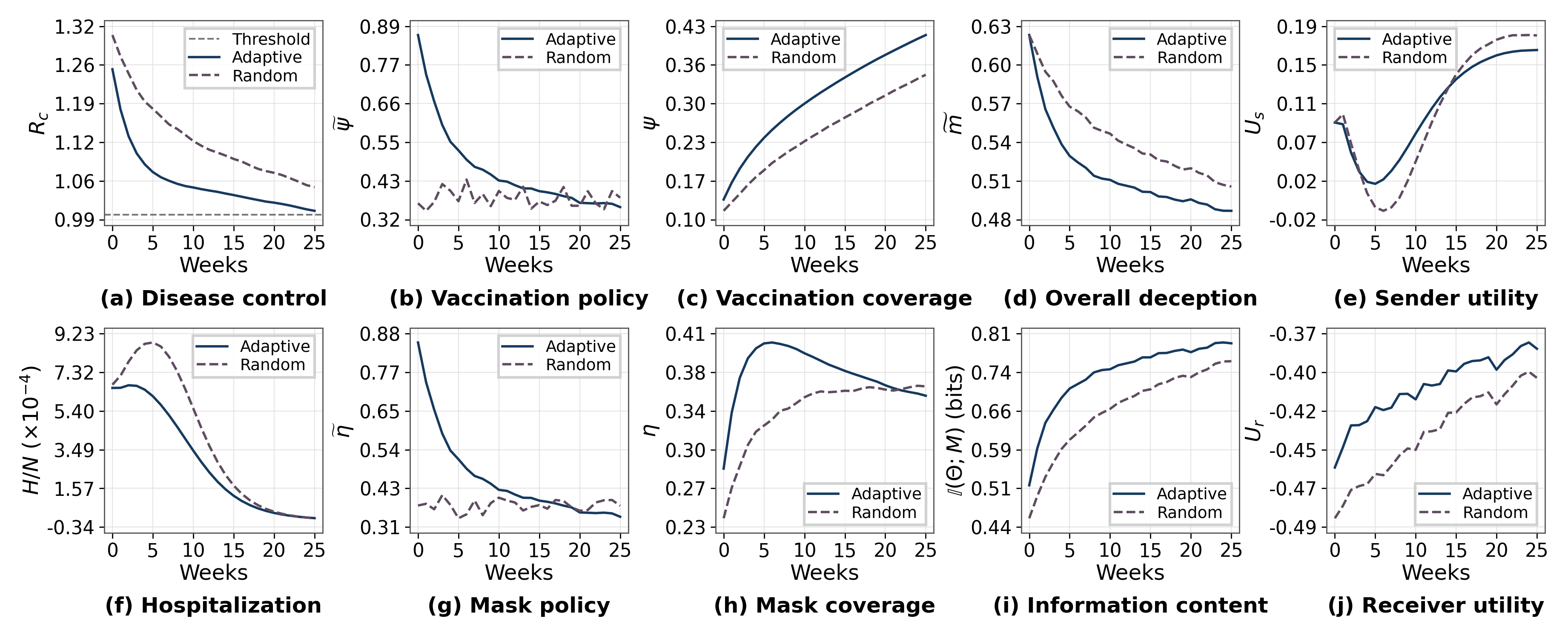}
    \caption{Comparison of adaptive and random policies under partial pooling with high initial behavioral rates, across epidemic dynamics, policy actions, signaling behavior, and utilities.}
    \label{fig:gov_policy_partial_high} %\vspace{-4mm}
\end{figure*}

\begin{figure*}[h]
    \centering
    \includegraphics[width=\textwidth, height=5cm]{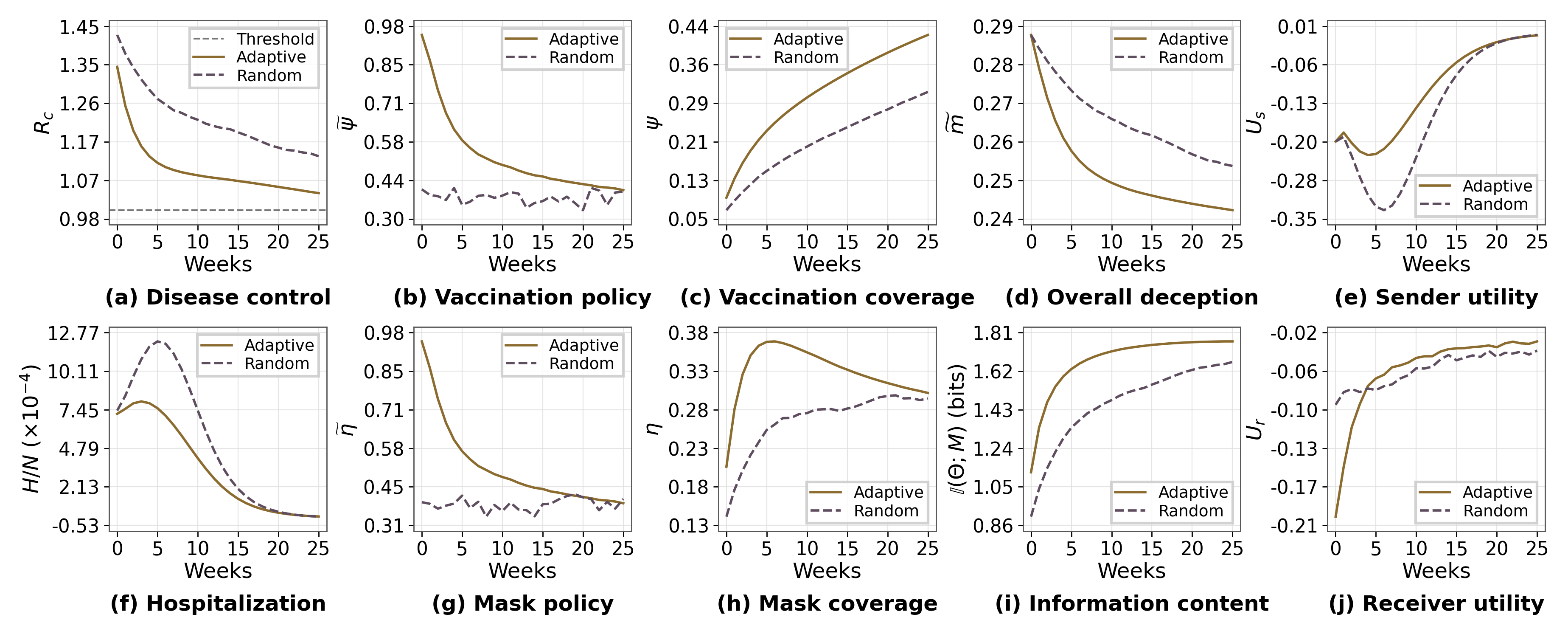}
    \caption{Comparison of adaptive and random policies under separating with low initial behavioral rates, across epidemic dynamics, policy actions, signaling behavior, and utilities.}
    \label{fig:gov_policy_partial_high} %\vspace{-4mm}
\end{figure*}

\begin{figure*}[h]
    \centering
    \includegraphics[width=\textwidth, height=5cm]{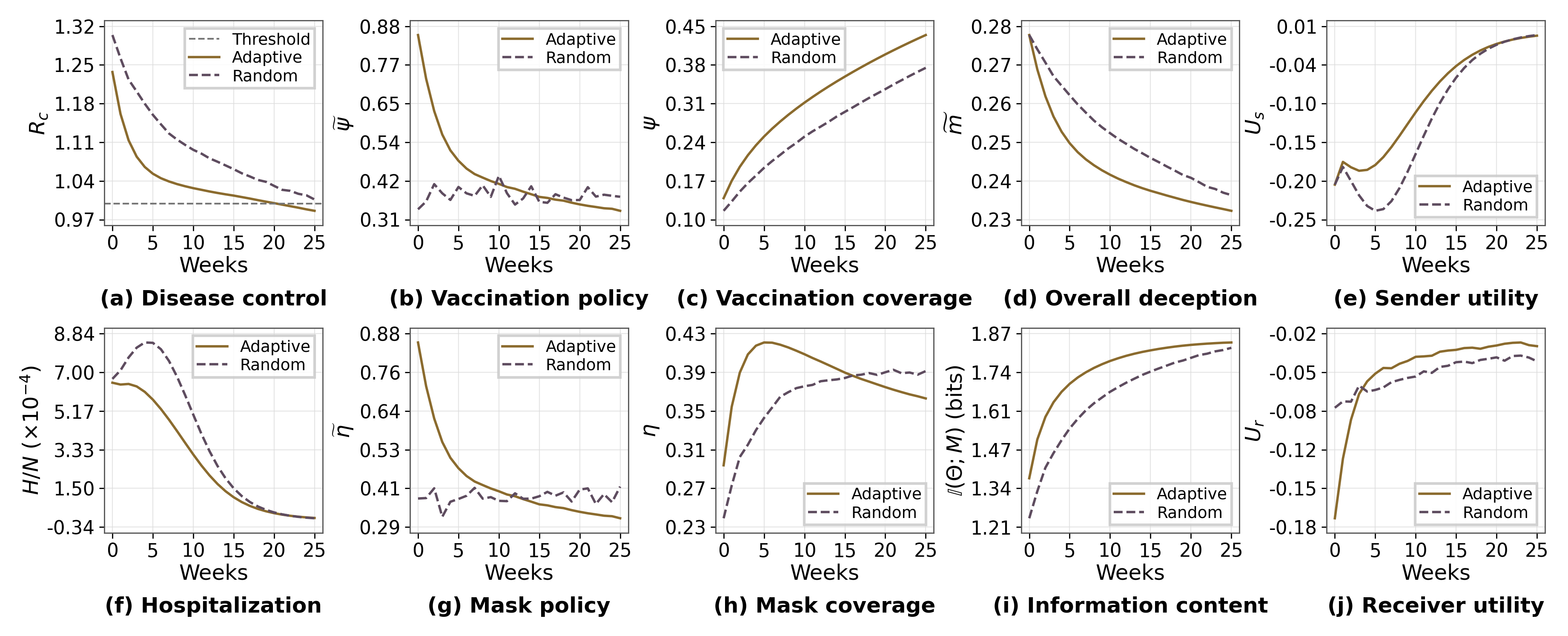}
    \caption{Comparison of adaptive and random policies under separating with high initial behavioral rates, across epidemic dynamics, policy actions, signaling behavior, and utilities.}
    \label{fig:gov_policy_partial_high} %\vspace{-4mm}
\end{figure*}

\begin{figure*}[h]
    \centering
    \includegraphics[width=\textwidth, height=5cm]{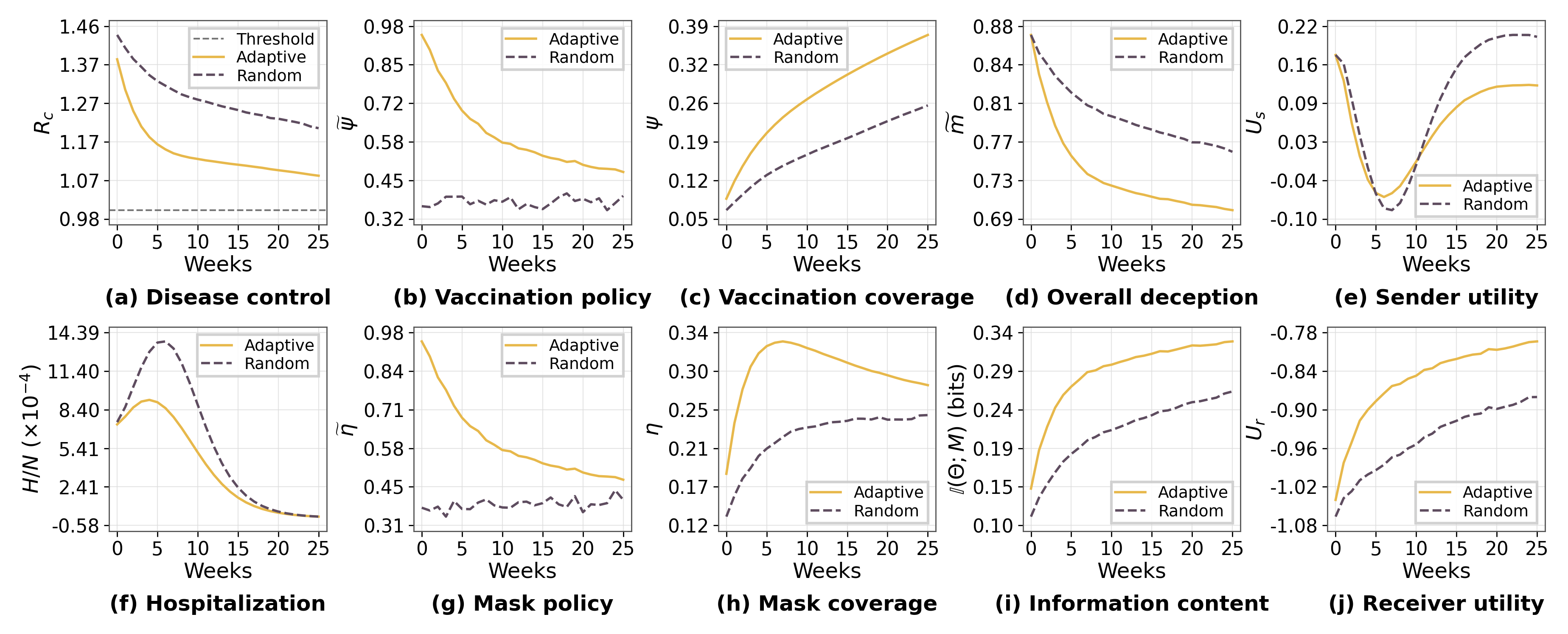}
    \caption{Comparison of adaptive and random policies under pooling with low initial behavioral rates, across epidemic dynamics, policy actions, signaling behavior, and utilities.}
    \label{fig:gov_policy_partial_high} %\vspace{-4mm}
\end{figure*}

\begin{figure*}[h]
    \centering
    \includegraphics[width=\textwidth, height=5cm]{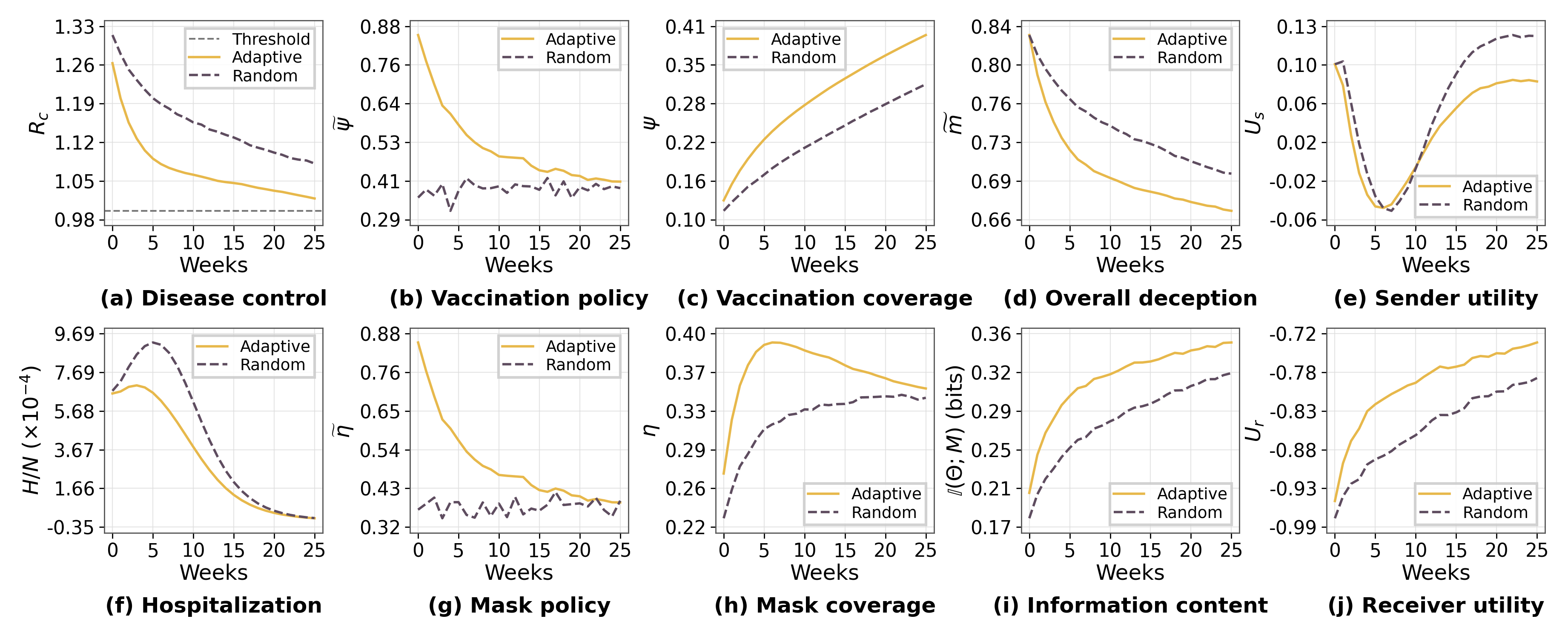}
    \caption{Comparison of adaptive and random policies under pooling with high initial behavioral rates, across epidemic dynamics, policy actions, signaling behavior, and utilities.}
    \label{fig:gov_policy_partial_high} %\vspace{-4mm}
\end{figure*}

% \bibliographystyle{ACM-Reference-Format}
% \bibliography{signaling_games}

\end{document}